\documentclass[12pt, draftclsnofoot, onecolumn]{IEEEtranTCOM}
\renewcommand{\baselinestretch}{1.4}

\usepackage{indentfirst}
\usepackage{multirow}
\usepackage{graphicx}
\usepackage{color}
\usepackage{xcolor}
\usepackage{colortbl}
\definecolor{mypink}{rgb}{.99,.91,.95}
\usepackage{listings}
\usepackage{enumerate}
\usepackage{subfigure}
\usepackage{amsmath,amsthm,amscd,amssymb,latexsym,upref,stmaryrd}
\usepackage[noend]{algpseudocode}
\usepackage{algorithmicx,algorithm}
\usepackage{amsfonts,hyperref,epsfig}
\usepackage{CJK}
\usepackage{float}
\usepackage{cite}
\usepackage{subfigure}
\usepackage{lipsum}
\usepackage{multicol}
\usepackage{setspace}
\usepackage{bm}
\usepackage{stfloats}
\usepackage[aboveskip=-0pt,belowskip=-30pt]{caption}
\usepackage{geometry}

\newtheorem{Theorem}{Theorem}
\newtheorem{Lem}{Lemma}%
\def\BibTeX{{\rm B\kern-.05em{\sc i\kern-.025em b}\kern-.08em
		T\kern-.1667em\lower.7ex\hbox{E}\kern-.125emX}}

\newtheorem{remark}{Remark}
\begin{document}
\title{3-D Positioning and Environment Mapping for mmWave Communication Systems}

\author{Jie~Yang,
	Shi~Jin,	Chao-Kai~Wen,
	Jiajia~Guo,
	and~Michail~Matthaiou
	\thanks{Jie~Yang, Shi~Jin, and Jiajia~Guo are with the National Mobile Communications Research Laboratory, Southeast University, Nanjing, China (e-mail: \{yangjie;jinshi;jiajiaguo\}@seu.edu.cn). Chao-Kai~Wen is with the Institute of Communications Engineering, National Sun Yat-sen University, Kaohsiung, 804, Taiwan (e-mail: chaokai.wen@mail.nsysu.edu.tw). Michail~Matthaiou is with the Institute of Electronics, Communications and Information Technology (ECIT), Queen's University Belfast, Belfast, U.K. (e-mail: m.matthaiou@qub.ac.uk).}
}

\maketitle	
\vspace{-1cm}
\begin{abstract}

Millimeter-wave (mmWave) cloud radio access networks (CRANs) provide new opportunities for accurate cooperative localization, in which large bandwidths, large antenna arrays, and increased densities of base stations allow for their unparalleled delay and angular resolution. Combining localization into communications and designing simultaneous localization and mapping (SLAM) algorithms are challenging problems. This study  considers the joint position and velocity estimation and environment mapping problem in a three-dimensional  mmWave CRAN architecture. We first embed cooperative localization into communications and establish the joint estimation and mapping model with hybrid delay and angle measurements. Then, we propose a closed-form weighted least square (WLS) solution for the joint estimation and mapping problems. The proposed WLS estimator is proven asymptotically unbiased and confirmed by simulations as effective in achieving the Cramer-Rao lower bound (CRLB) and the desired decimeter-level accuracy. Furthermore, we propose a WLS-Net-based SLAM algorithm by embedding neural networks into the proposed WLS estimators to replace the linear approximations. The solution possesses both powerful learning ability of the neural network and robustness of the proposed geometric model, and the  ensemble learning is applied to further improve positioning accuracy. A public ray-tracing dataset is used in the simulations to test the performance of the WLS-Net-based SLAM algorithm, which is proven fast and effective in attaining the centimeter-level accuracy.

\end{abstract}


\vspace{-0.4cm}
\section{Introduction}
\begin{spacing}{1.55}
Future networks should be able to offer unlimited coverage anywhere and anytime to any device to stimulate the amalgamation of positioning and wireless communications \cite{five}.	
Millimeter-wave (mmWave) communication is a promising technology to meet such requirements in future wireless communications. Localization is a highly desirable feature of mmWave communications \cite{loc0,locsum}. The location of the user equipment (UE) can be used to provide location-based services, such as navigation, mapping, social networking, augmented reality, and intelligent transportation systems, among others. Additionally, location-aware communication can be realized by the obtained location information to improve communication capacity and network efficiency \cite{loc1}.

mmWave bands offer larger bandwidths than the presently used sub-6 GHz bands, and higher resolution of the time delay (TD), time difference of arrival (TDoA), and frequency difference of arrival (FDoA) can be obtained accordingly. In addition, the penetration loss from mmWave bands is inherently large \cite{overview,Han2017}. Hence, the difference between the received power of the line-of-sight (LoS) path and the non-line-of-sight (NLoS) path is pronounced, which makes it much easier to eliminate the NLoS interference \cite{sparse}. As a manner of compensating severe penetration loss and increased path-loss, large antenna arrays and highly directional transmission should be combined \cite{lens} to facilitate the acquisition of the angle of arrival (AoA) and the angle of departure (AoD). Moreover, cloud radio access networks (CRANs) can enhance mmWave communication by improving the network coverage \cite{cran}. Therefore, mmWave CRANs can offer accurate cooperative localization in urban and indoor environments in which conventional GPS may fail. In return, the location information can improve the scalability, latency, and robustness of future networks \cite{loc}.


Localization has become a popular research topic in recent years. Different localization techniques have been summarized in \cite{loc2}.
The techniques are mainly divided into two categories: direct localization \cite{direct3,direct4} and indirect localization \cite{TD1,TD3,AOA,l2,l3,l4,l5}.
Direct localization converts directly the received waveform into a location estimation, but the process is highly complex.
By contrast, indirect localization applies the principle in which the channel parameters (AoA, TD, TDoA, and FDoA) are first extracted from the received waveform and grouped together as a function of the location parameters, and then different estimators are used to determine the UE positions.
In \cite{TD1,TD3}, several different closed-form TD-based algorithms have been proposed.
Few different AoA-based methods were developed in \cite{AOA} and in the related references.
AoA and its combination with ranging estimates are expected to achieve much higher position accuracy.
The works in \cite{l2} considered the positioning problem of three-dimensional (3-D) and stationary targets in MIMO radar systems that utilize hybrid TD/AoA measurements,
from which a novel computationally efficient closed-form algorithm is developed with the estimator to achieve the Cramer-Rao lower bound (CRLB) under small measurement noise.
The works of \cite{l3,l4} estimated the location and velocity of a moving target in radar systems.
The algorithm in \cite{l3} separately estimated position and velocity by introducing two hybrid pseudo-linear estimators and by using the hybrid TDoA and FDoA measurements of a signal received by different receivers.
The research in \cite{l4} dealt with a hybrid pseudo-linear estimator for the target motion analysis of a constant-velocity target in the two-dimensional (2-D) scenario by using AoA, TDoA, and FDoA measurements.
The work in \cite{l5} studied 3-D downlink positioning with a single reference station that requires both AoA and AoD measurements (obtained after beam training) and abundant single-bounce multipath components.

In this study, we consider the 3-D indirect localization of UE with variable velocities in mmWave systems by using hybrid AoA, TDoA, and FDoA measurements.
Channel parameters (e.g., AoA, TD, TDoA, and FDoA) can be measured much more accurately in the initial access and communication stages owing to the unparalleled delay and angular resolution of mmWave communications.
In particular, we focus on methods to improve positioning accuracy under the CRAN architecture.
CRANs provide a cost-effective way to achieve network densification \cite{co1,co2}, in which distributed low-complexity remote radio heads (RRHs) are deployed close to the users and coordinated by a central unit (CU) for joint processing.
The obtained location information can be shared with network nodes.
Therefore, cooperative localization is expected to improve the accuracy of the localization.


The LoS path is the most favorable path for localization and thus has been studied for a long time.
Extensive researches investigated the NLoS mitigation techniques used to reduce the positioning error \cite{nlos1,nlos2}.
However, a recent clarification in \cite{foe} has shown that future communication systems will turn multipath channels ``from foe to friend"  by leveraging distinguishable multipath components resulting from  unparalleled delay and angular resolution in mmWave systems.
Hence, the information from reflected signals can be exploited in the reconstruction of the 3-D map of the surrounding environment, also called environment mapping.
The simultaneous localization and mapping (SLAM) method, which has become popular in mmWave communications, was first proposed in robotics research \cite{slam}.
mmWave imaging and communication were both leveraged for SLAM in \cite{slam1} on the assumption of specular reflection for NLoS paths.
In \cite{slam3}, different types of reflections were considered (specular reflection, spread, diffusion, and their combination), however, the initial UE's position and the moving direction were needed.
Given that higher-order-reflected NLoS paths inherently tend to be affected considerably by attenuation at mmWave frequencies \cite{path,NLOS}, we only consider in this study the LoS path and the single-bounce NLoS paths.
In particular, we aim to establish an effective environment mapping method that does not require specular reflection assumption and initial position and moving direction of the UE.

All of the positioning techniques mentioned above are geometric approaches, where delay and angular measurements are extracted and from which the position and velocity of the UE, as well as the scatterers, are triangulated or trilaterated.
A function can be approximated by geometric techniques given the existence of an underlying transfer function between the measurements and the positions.
In recent years, artificial intelligence (AI) has received great attention because of its promising performance in solving complicated problems.
Researchers have utilized neural networks to learn underlying transfer functions, and AI-based positioning solutions, such as fingerprinting (FP) methods \cite{fingerprint0,fingerprint1}, have emerged.
A deep learning-based indoor FP system was presented in \cite{fingerprint0} to achieve meter-level positioning accuracy. The experiments in \cite{fingerprint1} showed the feasibility of using deep learning methods for localization in actual outdoor environments.
The FP methods have alleviated modeling issues; however, extremely large amounts of data are required to meet the high requirements of positioning accuracy.
To our best knowledge, the work on positioning by combining neural networks with geometric models are rare and thus the focus of the present study.

We address in this study the 3-D positioning problem of moving users in mmWave communication systems and environment mapping. The contributions of this study are as follows:

\begin{itemize}
\item \textbf{Joint Position and Velocity Estimation}:
We first establish a joint position and velocity estimation model by utilizing hybrid TDoA/FDoA/AoA measurements. Then, we develop an efficient closed-form weighted least square (WLS) algorithm. Unlike other closed-form WLS-based methods \cite{l3} with multistage estimators, the proposed method can determine the UE’s position and velocity in one stage only. The estimator is proven asymptotically unbiased and confirmed by simulations as effective in achieving CRLB under small measurement noise.

\item \textbf{Environment Mapping}:
We exploit the single-bounce NLoS paths and the estimated UE position to build the environment mapping model and locate the scatterers. Then, we deduce the closed-form WLS solution to determine the scatterer location. The ray-tracing dataset provided in \cite{deepmimo,raytracing} is used to verify the rationality of the single-bounce NLoS path assumption, and the proposed environment mapping algorithm is used to map the environment geometry that generates the ray-tracing dataset. The proposed environment mapping algorithm can achieve the CRLB.

\item \textbf{Neural Network-Assisted WLS Algorithm}:
We propose a neural network-assisted WLS algorithm named WLS-Net to further improve the positioning accuracy of localization and environment mapping. The algorithm benefits from the powerful learning ability of the neural network and the robustness of the geometric model.
In addition, WLS-Net is fast because it can eliminate iterations in the proposed WLS algorithm.
To our best knowledge, the present study is the first to combine a neural network and a geometric model in 3-D SLAM methods.
Furthermore, we embed the ensemble learning into the proposed WLS-Net, named eWLS-Net, to enhance localization accuracy.
The simulation results show that WLS-Net significantly outperforms the WLS algorithm when the measurement error has an intrinsic relationship.
\end{itemize}

The rest of this paper is organized as follows. In Section \uppercase\expandafter{\romannumeral2}, we introduce a 3-D mmWave system model and the relationship between the channel and the location parameters. The WLS-based joint position and velocity estimation algorithm is proposed in Section \uppercase\expandafter{\romannumeral3}. In Section \uppercase\expandafter{\romannumeral4}, we establish the WLS-based environment mapping algorithm. In Section \uppercase\expandafter{\romannumeral5}, we embed the neural networks into the proposed WLS algorithms, and ensemble learning is applied. Our simulation results are presented in Section \uppercase\expandafter{\romannumeral6}. We conclude the study in Section \uppercase\expandafter{\romannumeral7}.
\end{spacing}

{\bf Notations}---In all aspects of this study, uppercase boldface $\mathbf{A}$ and lowercase boldface $\mathbf{a}$ denote matrices and vectors, respectively. For any matrix $\mathbf{A}$, the superscripts $\mathbf{A}^{-1}$ and $\mathbf{A}^{T}$ stand for inverse and transpose, respectively,
and $\mathbf{A}(i,:)$ represents the $i$-th row of $\mathbf{A}$.
For any vector $\mathbf{a}$, the 2-norm is denoted by $\|\mathbf{a}\|$, $\mathbf{a}(i)$ represents the $i$-th entry in $\mathbf{a}$, and $\mathbf{a}(i:j)$ represents the $i$-th to $j$-th entries in $\mathbf{a}$.
$\mathbf{0}$ denotes an all-zero vector with a given size.
$\text{diag}\{\cdot\}$ denotes a diagonal matrix with entries in $\{\cdot\}$,
and
$\mbox{blkdiag}(\mathbf{A}_1,\mathbf{A}_2,\ldots,\mathbf{A}_k)$ denotes a block-diagonal matrix constructed
by $\mathbf{A}_1$, $\mathbf{A}_2$, and $\mathbf{A}_k$.
$\mathbb{E}\{\cdot\}$ denotes statistical expectation.
The notation $a^\circ$ is the true value of the estimated parameter $a$, while
 $\doteq$ denotes ``approximately equal".
 For any complex value $c$, $|c|$ denotes the module of $c$.

\vspace{-0.25cm}
\section{System Model}
We study the localization (joint position and velocity estimation of the moving user) and environment mapping problem in mmWave CRAN with $N$ RRHs \cite{cran},
as shown in Fig. \ref{fig:model}.
Each RRH is equipped with a large antenna array with $M$ antenna elements and connected to the CU via an individual fronthaul link with finite capacity. The CRAN has $K$ single omni-antenna UEs, in which the number of users is not greater than the total number of RRHs, i.e., $K\leq N$.
\begin{figure}
	\centering
	\includegraphics[scale=0.55,angle=0]{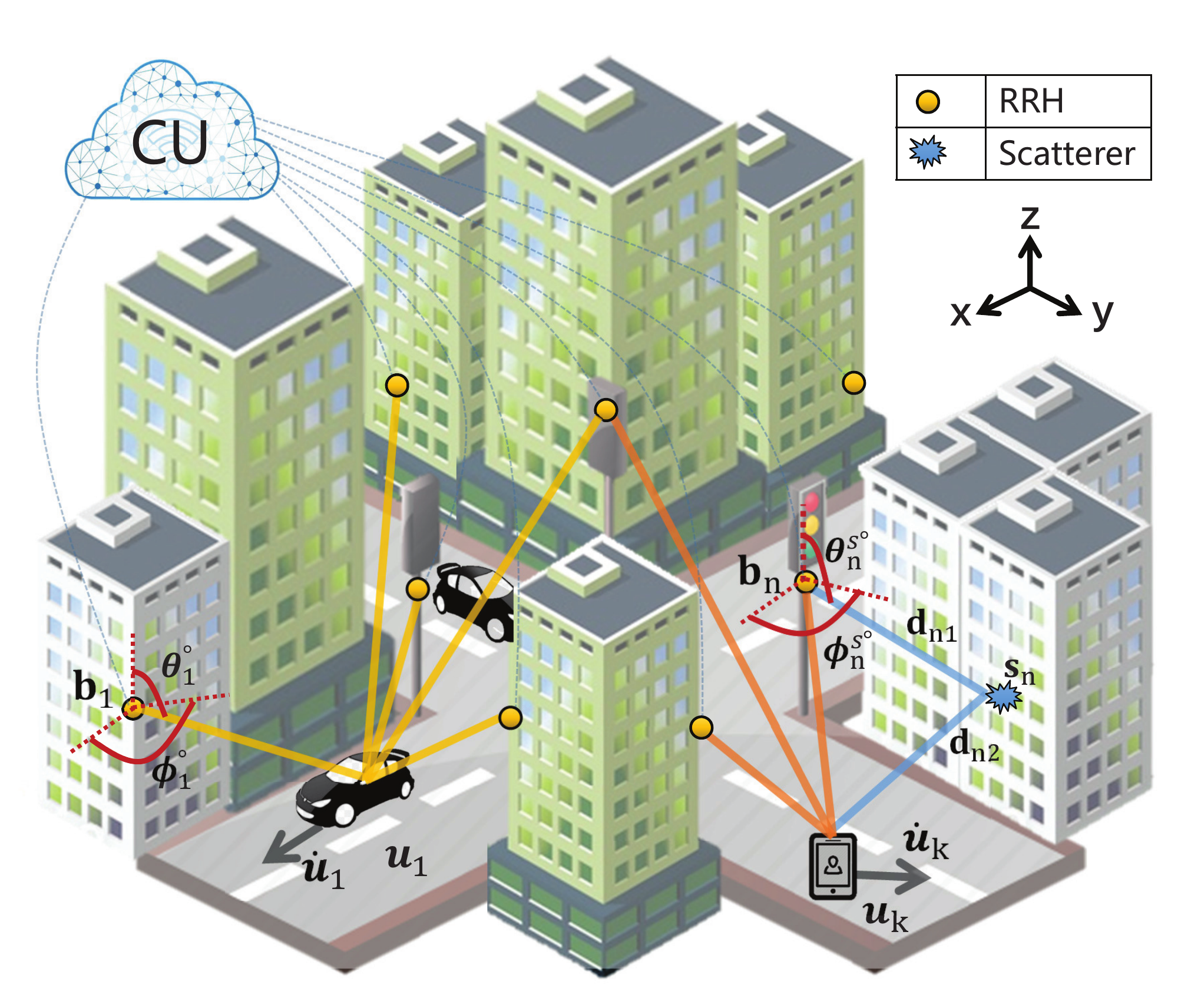}	
	\captionsetup{font=footnotesize}
	\caption{Illustration of the mmWave CRAN system model in which RRHs are connected with the CU.}
	\label{fig:model}	
\end{figure}

\vspace{-0.25cm}
\subsection{System Geometry}
We consider a 3-D space $\mathbb{R}^3=\{[x, y, z]^T: x, y, z\in\mathbb{R}\}$ with $N$ known RRHs located at $\mathbf{b}_n=[x_n^b, y_n^b, z_n^b]^T$, such that $n=1, 2, \ldots, N$.
The geometry between the RRHs and the UEs is shown in Fig. \ref{fig:model}.
We assume that the unknown position and velocity of $K$ UEs
are represented by $\mathbf{u}_k = [x_{k}, y_{k}, z_{k}]^T$ and ${\dot{\mathbf{u}}_k}=[\dot{x}_{k}, \dot{y}_{k}, \dot{z}_{k}]^T$, respectively, with $k=1, 2, \ldots, K$.
Due to the sparsity of the mmWave channel, we assume that only the LoS path and the single-bounce NLoS path exist\cite{sparse}.
For ease of description, we assume that an unknown scatterer\footnote{We assume that one scatterer contributes to one NLoS path, the proposed positioning and mapping algorithm is also feasible for more than one scatterers.} exists between the $n$-th RRH and the $k$-th UE, as represented by $\mathbf{s}_{n,k} = [x_{n,k}^s, y_{n,k}^s, z_{n,k}^s]^T$.
Our goal is to determine $\mathbf{u}_k$, $\dot{\mathbf{u}}_k$, and $\mathbf{s}_{n,k}$, where $n=1, 2, \ldots, N$ and $k=1, 2, \ldots, K$, by the signals received at the RRHs.

\vspace{-0.5cm}
\subsection{Channel Model}
UEs transmit over the mmWave wideband frequency-selective channel.
The channel between UE $k$ and RRH $n$ has $C_{n,k}$ clusters containing $L_{n,k,c}$ paths each.
Given that the mmWave channel is sparse, we assume that $C_{n,k}L_{n,k,c}\leq Q$, where $Q$ is the number of RF chains for each RRH.
Hence, the channel coefficient vector $\mathbf{h}_{n,k}(t)\in \mathbb{C}^{M\times1}$ at time $t$ between UE $k$ and RRH $n$ can be expressed by using the clustered channel model \cite{c-channel} as
\begin{equation}\label{channel}
\mathbf{h}_{n,k}(t) =  \sum_{c=0}^{C_{n,k}-1}\sum_{l=0}^{L_{n,k,c}-1}\alpha_{n,k,c,l} \delta(t-\tau_{n,k,c,l})\mathbf{a}(\phi_{n,k,c,l},\theta_{n,k,c,l}),
\end{equation}
where $\alpha_{n,k,c,l}$, $\tau_{n,k,c,l}$, $\phi_{n,k,c,l}$, and $\theta_{n,k,c,l}$ denote the complex gain, delay, azimuth AoA, and elevation AoA for the $l$-th path of the $c$-th cluster, respectively, and $\mathbf{a}(\cdot)$ is the array response vector at the RRH.

\vspace{-0.4cm}
\subsection{Transmission Model}\label{ra}
Positioning and mapping can be embedded in either the initial access stage or data transmission stage without additional overhead.
Take the initial access \cite{initial} as an example. Each  RRH periodically transmits synchronization signals to the UEs. After detecting the synchronization signals and decoding the broadcast messages, the UE randomly selects one of the small numbers (at most 64 in LTE) of the waveforms, also called random access (RA) preambles, and transmits it in one of the RA slots.
Hence, we can assume that there is no need to consider the uplink interference from the other UEs due to the near-orthogonality state in terms of waveform and time. Subsequently, we take one UE to simplify the expressions and thus omit the $k$ index.

The UE omni-directionally sends a signal $\sqrt{E_s}s(t)$,\footnote{The major drawback of the uplink positioning comes from the total transmit power limit of the UE, which can be compensated by multi-antennas and beamforming at the UE side.} in which $E_s$ is the transmitted energy, and $\mathbb{E}\{|s(t)|^2\!\}\!=\!1$.
The received signal $\mathbf{r}_{n}(t) \in \mathbb{C}^{Q\times 1}$ at RRH $n$ is given by
\begin{equation}\label{receive}
\mathbf{r}_{n}(t)= \mathbf{A} \sum_{c=0}^{C_{n}-1}\sum_{l=0}^{L_{n,c}-1}\alpha_{n,c,l} \sqrt{E_s}s(t-\tau_{n,c,l})\mathbf{a}(\phi_{n,c,l},\theta_{n,c,l})+\mathbf{n}(t),
\end{equation}
where ($\phi_{n,c,l}$, $\theta_{n,c,l}$, $\tau_{n,c,l}$) are the channel parameters for different paths, $\mathbf{A} \in \mathbb{R}^{Q\times M}$ is the RF switch network in the mmWave hybrid beamforming architecture, and $\mathbf{n}(t)  \in \mathbb{C}^{Q\times 1}$ is the zero-mean white Gaussian noise with a known power spectrum density.
In this study, we focus on the proposed localization algorithm itself.
Then, all distinct paths with estimated AoAs and delays \footnote{The time delay is estimated by detecting the first peak of the correlation between the received signal and the transmitted reference signal, e.g. PRS.} in RRH $n$ are transmitted to the CU for $n = 1, \ldots, N$.
The LoS path can be identified easily at the mmWave frequencies because it is stronger than the NLoS paths.
We assume that a threshold should be selected such that the CU can generally select $N_{a}$ ($2 \leq N_{a} \leq N$) LoS paths from $N$ RRHs to ensure high localization accuracy.
After obtaining the channel parameters $(\phi_{n,c,l}, \theta_{n,c,l}, \tau_{n,c,l})$, the CU estimates the position and velocity of the UE based on the channel parameters of the LoS paths and maps the environment with the NLoS paths.

\vspace{-0.25cm}
\subsection{Relation Between Channel and Location Parameters}
The hybrid measurements we use in this study are TDoA, FDoA, and AoA, which are related to the channel parameters as follows: delays $\tau_{n,c,l}$ and AoA pairs $(\phi_{n,c,l},\theta_{n,c,l})$ for distinct paths
obtained from \eqref{receive}.
Then, we map these channel parameters to the location parameters.
For the sake of presentation clarity,
the unknown position and velocity of the UE are $\mathbf{u}^\circ=[x^\circ, y^\circ, z^\circ]^T$ and $\dot{\mathbf{u}}^\circ=[\dot{x}^\circ, \dot{y}^\circ, \dot{z}^\circ]^T$, respectively.
Moreover, we assume that only one scatterer $\mathbf{s}_{n}^\circ = [x_{n}^{s\circ}, y_{n}^{s\circ}, z_{n}^{s\circ}]^T$ and one NLoS path exist between the RRH $n$ and the UE. Hence, $(c=0,l=0)$ can be denoted as the LoS path and $(c=1,l=1)$ can be denoted as  the NLoS path.
\begin{itemize}
\item \textbf{TDoAs}:
For the LoS path, the ToA is given by
$
\tau_{n,0,0} = ||\mathbf{u}^\circ-\mathbf{b}_n||/v_c + \omega,
$
where $v_c$ is the speed of light, and $\omega$ is the unknown clock bias between CRAN and UE. The \textbf{TDoA} of a signal received by the RRH $n$ and $1$ is
$\Delta\tau_{n} = \tau_{n,0,0} - \tau_{1,0,0}$. Therefore, the unknown $\omega$ can be eliminated.
Let
\begin{eqnarray} \label{2}
r_{n1}^\circ=\Delta\tau_{n}\times v_c = ||\mathbf{u}^\circ-\mathbf{b}_n|| - ||\mathbf{u}^\circ-\mathbf{b}_1||,
\end{eqnarray}
and we define
\begin{eqnarray} \label{1}
r_n^\circ=||\mathbf{u}^\circ-\mathbf{b}_n||=\sqrt{(\mathbf{u}^\circ-\mathbf{b}_n)^T (\mathbf{u}^\circ-\mathbf{b}_n)}.
\end{eqnarray}
For the NLoS path, we have
$
\tau_{n,1,1} = ||\mathbf{u}^\circ-\mathbf{s}_n^\circ||/v_c + ||\mathbf{s}_n^\circ-\mathbf{b}_n||/v_c + \omega.
$
The TDoA of the single-bounce NLoS path received by the RRH $n$ and reference time (the delay of LoS path received by the RRH $1$)
is
$\Delta\tau_{n}^s = \tau_{n,1,1} - \tau_{1,0,0}$.
Let
\begin{eqnarray} \label{sc1}
r_{n1}^{s\circ} \! = \! \Delta\tau_{n}^s \!\times\! v_c \! = \! ||\mathbf{u}^\circ\!-\!\mathbf{s}_n^\circ|| \! + \! ||\mathbf{s}_n^\circ\!-\!\mathbf{b}_n||\! -\! ||\mathbf{u}^\circ\!-\!\mathbf{b}_1||.
\end{eqnarray}
Hence, $r_{n1}^\circ$ and $r_{n1}^{s\circ}$ are the TDoA-related parameters, which will be used in our proposed algorithms, and they are derived from the TDoA by multiplying with $v_c$.

\item \textbf{FDoAs}:
FDoA is the change rate of TDoA with time. The time derivative of ToA is FoA, namely, Doppler shift.
The time derivative of $r_n^\circ$ in \eqref{1} is given by
\begin{equation}\label{3}
\hspace{-0.3cm}\dot{r}_{n}^\circ \!=\!
 \dfrac{\dot{\mathbf{u}}^{\circ T}\mathbf{u}^\circ+\mathbf{u}^{\circ T}\dot{\mathbf{u}}^\circ-2\dot{\mathbf{u}}^{\circ T}\mathbf{b}_n}{2\sqrt{(\mathbf{u}^\circ-\mathbf{b}_n)^T (\mathbf{u}^\circ-\mathbf{b}_n)}}\!=\!\dfrac{\dot{\mathbf{u}}^{\circ T}(\mathbf{u}^\circ-\mathbf{b}_n)}{r_n^\circ}.
\end{equation}
We define
\vspace{-0.25cm}
\begin{equation} \label{4}
\dot{r}_{n1}^\circ=\dot{r}_n^\circ-\dot{r}_1^\circ,
\vspace{-0.25cm}
\end{equation}
where $\dot{r}_{n1}^\circ$ denotes the FDoA-related parameters derived from the FDoA by multiplying with $v_c$.
The FDoA is used to estimate the velocity of the UE and hence illustrated only for the LoS path.

\item \textbf{AoAs}:
For the LoS path,
\vspace{-0.25cm}
\begin{equation} \label{5}
\phi_{n}^\circ=\phi_{n,0,0}^\circ= \arctan\frac{y^\circ-y_n^b}{x^\circ-x_n^b},  \ \ \ \ \ \theta_{n}^\circ=\theta_{n,0,0}^\circ=\arcsin\frac{z^\circ-z_n^b}{||\mathbf{u}^\circ-\mathbf{b}_n||},
\end{equation}
holds. Then, for the NLoS path, we have
\vspace{-0.25cm}
\begin{equation} \label{55}
\phi_n^{s\circ} = \phi_{n,1,1}^\circ= \arctan\frac{y_{s,n}^\circ-y_n^b}{x_{s,n}^\circ-x_n^b} ,  \ \ \ \theta_n^{s\circ} = \theta_{n,1,1}^\circ=\arcsin\frac{z_{s,n}^\circ-z_n^b}{||\mathbf{s}_n^\circ-\mathbf{b}_n||}.\\
\end{equation}
\end{itemize}
The hybrid measurements (TDoA, FDoA, and AoA) will be utilized to estimate the position and velocity of the UE and map the propagation environment.
Given that the relationships \eqref{2}-\eqref{55} are nonlinear and nonconvex functions of $\mathbf{u}^\circ$, $\dot{\mathbf{u}}^\circ$, or $\mathbf{s}_n$, for $n=1, 2, \ldots, N$,  solving the positioning and mapping problem is not a trivial task.
In the subsequent sections,
effective positioning and mapping estimators will be proposed.

\vspace{-0.2cm}
\section{Joint Position and Velocity Estimation}\label{wls1}
In this section, we propose an effective localization method to simultaneously predict the unknown position $\mathbf{u}^\circ$ and velocity $\dot{\mathbf{u}}^\circ$ of the UE.

\vspace{-0.3cm}
\subsection{Joint Position and Velocity Estimation Model}
We first obtain the matrix representation of the joint position and velocity estimation model.
Let{\begingroup\makeatletter\def\f@size{12}\check@mathfonts		 \def\maketag@@@#1{\hbox{\m@th\normalsize\normalfont#1}}\setlength{\arraycolsep}{0.0em}\setlength{\arraycolsep}{0.0em}
	\begin{eqnarray}\label{angle}
	\begin{split}
	&\mathbf{a}_n^\circ=[\cos\theta_n^\circ\cos\phi_n^\circ, \cos\theta_n^\circ\sin\phi_n^\circ, \sin\theta_n^\circ]^T, \ \ \ \ \ \ \ \  \mathbf{c}_n^\circ=[-\sin\phi_n^\circ, \cos\phi_n^\circ, 0]^T,\\
	&\mathbf{d}_n^\circ=[-\sin\theta_n^\circ\cos\phi_n^\circ, -\sin\theta_n^\circ\sin\phi_n^\circ, \cos\theta_n^\circ]^T,
	\end{split}
	\end{eqnarray}\setlength{\arraycolsep}{5pt}\endgroup}for $n=1,\ldots,N_{a}$. By applying the nonlinear transformation to \eqref{2},\eqref{1},\eqref{3}-\eqref{5} and by utilizing the angular relationships in \eqref{angle}, we obtain the noise-free matrix representation of the joint position and velocity estimation model, which is given by
\vspace{-0.25cm}
\begin{equation} \label{191}
\mathbf{h}=\mathbf{G}\mathbf{x}^\circ,
\vspace{-0.25cm}
\end{equation}
where $\mathbf{x}^\circ=[\mathbf{u}^{\circ T}, \dot{\mathbf{u}}^{\circ T}]^T$ is the vector of unknown position and velocity of the UE,\\ $\mathbf{h}=[\mathbf{q}_2^T,\dots,\mathbf{q}_{N_{a}}^T, \mathbf{h}_1^T, \dots,\mathbf{h}_{N_{a}}^T]^T$, $\mathbf{G}=[\mathbf{P}_2^T,\dots,\mathbf{P}_{N_{a}}^T, \mathbf{G}_1^T, \dots,\mathbf{G}_{N_{a}}^T]^T$, and
\begin{equation*}
\begin{split}
&\mathbf{q}_n\!=\!\begin{pmatrix}
(r_{n1}^\circ)^2\!-\!2r_{n1}^\circ \mathbf{a}_1^{\circ T}\mathbf{b}_1\!-\!\mathbf{b}_n^T \mathbf{b}_n\!+\!\mathbf{b}_1^T \mathbf{b}_1 \\
\dot{r}_{n1}^{\circ}{r}_{n1}^{\circ}-\dot{r}_{n1}^{\circ}\mathbf{a}_1^{\circ T}\mathbf{b}_1\\
\end{pmatrix},\ \ \ \ \ \ \ \ \ \ \ \ \ \mathbf{h}_n\!=\!\begin{pmatrix}
\mathbf{c}_n^{\circ T}\!\mathbf{b}_n \\
\mathbf{d}_n^{\circ T}\!\mathbf{b}_n
\end{pmatrix},\\
\end{split}
\end{equation*}
\begin{equation*}
\begin{split}
&\mathbf{P}_n\!=\!\begin{pmatrix}\!
2[(\mathbf{b}_1\!-\!\mathbf{b}_n)^{T}\!-\!r_{n1}^{\circ}\mathbf{a}_1^{\circ T}] & \mathbf{0}\\
-\dot{r}_{n1}^{\circ}\mathbf{a}_1^{\circ T} & (\!\mathbf{b}_1\!-\!\mathbf{b}_n\!)^{T}\!-\!r_{n1}^{\circ}\mathbf{a}_1^{\circ T}\!
\end{pmatrix}\!, \ \  \mathbf{G}_n\!=\!\begin{pmatrix}
\mathbf{c}_n^{\circ T} \!&\! \mathbf{0}\\
\mathbf{d}_n^{\circ T} \!&\! \mathbf{0}
\end{pmatrix}\!,
\end{split}
\end{equation*}
and $\mathbf{0}$ is a $1\times 3$ zero vector.
The detailed derivations of \eqref{191} are listed in Appendix \ref{A}.

The elements in vector $\mathbf{h}$ and matrix $\mathbf{G}$ in (\ref{191}) are operations of the noise-free channel parameters.
Let the measured parameters with noise replace the true parameters in (\ref{191}) (i.e., let $r_{i1}\!=\!r_{i1}^{\circ}\!+\!\Delta r_{i1}$, $\dot{r}_{i1}=\dot{r}_{i1}^{\circ}+\Delta \dot{r}_{i1}$, $\phi_j=\phi_j^{\circ}+\Delta \phi_j$, and
$\theta_j=\theta_j^{\circ}+\Delta\theta_j$ replace $r_{i1}^{\circ}$, $\dot{r}_{i1}^{\circ}$, $\phi_j^{\circ}$, and $\theta_j^{\circ}$ in $\mathbf{h}$ and $\mathbf{G}$) for $i=2,\ldots,N_{a}$ and $j=1,\ldots,N_{a}$. Then, we derive
\vspace{-0.25cm}
\begin{equation}\label{20}
\tilde{\mathbf{h}}=\tilde{\mathbf{G}}\mathbf{x}^{\circ}+\mathbf{e},
\vspace{-0.25cm}
\end{equation}
where $\tilde{\mathbf{h}}$ and $\tilde{\mathbf{G}}$ are the measured counterparts, and $\mathbf{e}$ is the error vector.
Equation \eqref{20} represents the pseudo-linear relationship between the measured channel parameters and unknown location and velocity parameters.

\vspace{-0.2cm}
\subsection{WLS-Based Joint Estimation Algorithm}\label{wls}
By relying on the proposed model in \eqref{20}, we can obtain the WLS
estimation of $\mathbf{x}^{\circ}$ according to \cite{KAY}. Denote
\begin{equation}
\mathbf{m}=[r_{21}, \dot{r}_{21}, \ldots, r_{{N_{a}}1}, \dot{r}_{{N_{a}}1},\phi_1, \theta_1,\ldots,\phi_{N_{a}}, \theta_{N_{a}}]^T,
\end{equation}
\begin{equation}\label{mo}
\mathbf{m}^\circ=[r_{21}^\circ, \dot{r}_{21}^\circ, \ldots, r_{{N_{a}}1}^\circ, \dot{r}_{{N_{a}}1}^\circ,\phi_1^\circ, \theta_1^\circ, \ldots,\phi_{N_{a}}^\circ, \theta_{N_{a}}^\circ]^T,
\vspace{-0.25cm}
\end{equation}
and
\begin{equation}\label{m}
\Delta \mathbf{m}=[\Delta r_{21}, \Delta\dot{r}_{21}, \ldots, \Delta r_{N_{a}1}, \Delta\dot{r}_{N_{a}1}, \Delta\phi_1, \Delta\theta_1, \ldots, \Delta\phi_{N_{a}}, \Delta\theta_{N_{a}}]^T,
\vspace{-0.25cm}
\end{equation}
and then we have
\vspace{-0.25cm}
\begin{eqnarray}\label{m0}
\mathbf{m}=\mathbf{m}^\circ+\Delta \mathbf{m},
\end{eqnarray}
where $\mathbf{m}$ is the vector of measured parameters, $\mathbf{m}^\circ$ is the vector of true parameters, and $\Delta \mathbf{m}$ is the vector of noise terms.
We assume that $\Delta \mathbf{m}$ is a zero-mean Gaussian vector, whose covariance matrix is $\mathbf{Q}$.
\begin{Lem}\label{lemma}
	The WLS estimation of  $\mathbf{x}^{\circ}$ is given by
	\vspace{-0.15cm}
	\begin{equation}\label{21}
	\mathbf{x}=(\tilde{\mathbf{G}}^{T}\mathbf{W}\tilde{\mathbf{G}})^{-1}\tilde{\mathbf{G}}^{T}\mathbf{W}\tilde{\mathbf{h}},
	\vspace{-0.25cm}
	\end{equation}
	where the positive definite weighting matrix $\mathbf{W}$ is taken into account to minimize the Frobenius norm of the covariance matrix of $\mathbf{x}$.
	If the covariance matrix $\mathbf{Q}$ is known,
	the weighting matrix can be deduced as follows:
	\vspace{-0.15cm}
	\begin{equation}\label{31}
	\mathbf{W}=(\mathbf{B}\mathbf{Q}\mathbf{B}^T)^{-1},
	\vspace{-0.25cm}
	\end{equation}
	where $\mathbf{B}$ (see \eqref{b} in Appendix \ref{B}) is the coefficient matrix with $\mathbf{B}\Delta \mathbf{m}$ approximating the linear term of $\mathbf{e}$, i.e., $\mathbf{e} \doteq \mathbf{B}\Delta \mathbf{m}$.
	Hence, the WLS estimation is given by
	\vspace{-0.15cm}
	\begin{equation}\label{331}
	\mathbf{x} = \left(\tilde{\mathbf{G}}^{T}(\mathbf{B}\mathbf{Q}\mathbf{B}^T)^{-1}\tilde{\mathbf{G}}\right)^{-1}\tilde{\mathbf{G}}^{T}(\mathbf{B}\mathbf{Q}\mathbf{B}^T)^{-1}\tilde{\mathbf{h}}.
	\vspace{-0.25cm}
	\end{equation}
\end{Lem}
\begin{proof}
	See Appendix \ref{B}.
\end{proof}
The proposed algorithm is summarized in Algorithm \ref{alg1},
in which $T$ represents the total number of iterations.
The weighting matrix $\mathbf{W}$ in \eqref{31} is dependent on the unknown location $\mathbf{u}^\circ$ and velocity $\dot{\mathbf{u}}^\circ$ via the matrix $\mathbf{B}$ defined in \eqref{b}.
Hence, in line 1 of Algorithm \ref{alg1}, we initialize $\mathbf{W}$ with $\mathbf{Q}^{-1}$. Then, in lines 2 and 3, we obtain the coarse estimation $\mathbf{x}$. 	
By updating $\mathbf{W}$ according to lines 5-7, we can produce a highly accurate solution of $\mathbf{x}$ in line 8.
\vspace{-0.3cm}
\begin{algorithm}
	\caption{\textbf{: Joint Position and Velocity Estimation}}\label{alg1}
	\begin{algorithmic}[1] 
		\Require $\mathbf{Q}$, $\mathbf{b}_j$, and measurements set $\mathbb{S}\!=\!\{r_{i1},\dot{r}_{i1},\phi_j,\theta_j\}$,
	for $i\!=\!2,\ldots,N_{a}$ and $j\!=\!1,\ldots,N_{a}$.
		\Ensure $\mathbf{u}=\mathbf{x}(1:3)$, $\dot{\mathbf{u}}=\mathbf{x}(4:6)$
		\State Initialization: $t=1$, $\mathbf{W}=\mathbf{Q}^{-1}$.
		\State Calculate $\tilde{\mathbf{h}}$ and $\tilde{\mathbf{G}}$ by the given measurements in $\mathbb{S}$.
		\State Calculate the initial $\mathbf{x}$ in (\ref{21}) with $\mathbf{W}=\mathbf{Q}^{-1}$.
		\While {$t\leq T$}
		\State Calculate $r_{i1}$, $r_j$, $\dot{r}_j$, $\mathbf{c}_1^\circ$, $\mathbf{d}_1^\circ$, $\dot{\phi}_1$, $\dot{\theta}_1$ by the obtained $\mathbf{x}$, for $i\!=\!2,\ldots,N_{a}$ and $j\!=\!1,\ldots,N_{a}$, \hspace*{0.2in} according to
		 \eqref{2}, \eqref{1}, \eqref{3}, \eqref{angle} and \eqref{6}.
		\State Generate $\mathbf{B}$ in (\ref{b}) by parameters obtained in step $5$.
		\State Update weighting matrix $\mathbf{W}$ in (\ref{31}) by obtained $\mathbf{B}$.
		\State Update $\mathbf{x}$ according to (\ref{21}).
		\State $t=t+1$.
		\EndWhile
		\State \Return{$\mathbf{x}$}
	\end{algorithmic}
\end{algorithm}

\vspace{-0.4cm}
\subsection{Performance Analysis}
In this section, we evaluate the performance of the proposed estimator in (\ref{331}) by analyzing the expectation and covariance of the proposed estimator.

\textbf{$\bullet$} \textbf{Asymptotic unbiasedness} occurs when the expectation of an estimator equals to the true value of the estimated parameters.
\vspace{-0.25cm}
\begin{Theorem}\label{theorem1}
	The presented estimator $  \mathbf{x}=(\tilde{\mathbf{G}}^{T}\mathbf{W}\tilde{\mathbf{G}})^{-1}\tilde{\mathbf{G}}^{T}\mathbf{W}\tilde{\mathbf{h}}$ with $ \mathbf{W}=(\mathbf{B}\mathbf{Q}\mathbf{B}^T)^{-1}$ is asymptotically unbiased, i.e., $\mathbb{E}\{\mathbf{x}\}\doteq \mathbf{x}^{\circ}$.
\end{Theorem}	
\begin{proof}
	The true value $\mathbf{x}^{\circ}$ of $\mathbf{x}$ can be expressed by $(\tilde{\mathbf{G}}^{T}\mathbf{W}\tilde{\mathbf{G}})^{-1}(\tilde{\mathbf{G}}^{T}\mathbf{W}\tilde{\mathbf{G}})\mathbf{x}^{\circ}$.
	Then, from (\ref{21}), we can obtain
	\vspace{-0.25cm}
	\begin{equation}\label{32}
	\begin{aligned}
	\Delta\mathbf{x}&=\mathbf{x}-\mathbf{x}^{\circ}
	=(\tilde{\mathbf{G}}^{T}\mathbf{W}\tilde{\mathbf{G}})^{-1}\tilde{\mathbf{G}}^{T}\mathbf{W}\mathbf{e}.
	\end{aligned}
	\vspace{-0.25cm}
	\end{equation}
	Taking the expectation in (\ref{32}) indicates that
	\vspace{-0.25cm}
	\begin{equation}
	\mathbb{E}\{\Delta\mathbf{x}\}=\mathbb{E}\{\mathbf{x}\}-\mathbf{x}^{\circ}=(\tilde{\mathbf{G}}^{T}\mathbf{W}\tilde{\mathbf{G}})^{-1}\tilde{\mathbf{G}}^{T}\mathbf{W}\mathbb{E}\{\mathbf{e}\}.
	\vspace{-0.25cm}
	\end{equation}
	Given $\mathbb{E}\{\mathbf{e}\}\doteq \mathbb{E}\{ \mathbf{B}\Delta \mathbf{m}\} = \mathbf{B} \mathbb{E}\{ \Delta \mathbf{m}\}=0$,
	we have $\mathbb{E}\{\Delta\mathbf{x}\}\doteq 0$. Therefore,  $\mathbb{E}\{\mathbf{x}\}\doteq \mathbf{x}^{\circ}$.
\end{proof}
\vspace{-0.25cm}
\begin{remark}
We have proven that the proposed estimator is asymptotically unbiased, which means that the proposed algorithm will become increasingly accurate as the number of measurements increases.
\end{remark}
\vspace{-0.25cm}
\textbf{$\bullet$} \textbf{CRLB} is the lowest possible variance that an unbiased linear estimator can achieve.
According to \cite{KAY}, the CRLB of $\mathbf{x}^{\circ}$ for the Gaussian noise model can be defined as
\vspace{-0.25cm}
\begin{equation}\label{b1}
\mbox{\mbox{CRLB}}(\mathbf{x}^{\circ})=(\mathbf{B}_1^T\mathbf{Q}^{-1}\mathbf{B}_1)^{-1},
\vspace{-0.25cm}
\end{equation}
where $\mathbf{B}_1=\partial \mathbf{m}^\circ/ \partial {\mathbf{x}^{\circ T}}$,
and $\mathbf{m}^\circ$ is defined in \eqref{mo}.
The partial derivatives are given by
	\begin{equation}\label{335}	
	\hspace{-0.25cm}
	\frac{\partial \mathbf{m}^\circ}{\partial\mathbf{x}^{\circ T}}
	\!\!=\!\!\left[\!(\frac{\partial r_{21}^{\circ }}{\partial\mathbf{x}^{\circ T}})^T,
	(\frac{\partial \dot{r}_{21}^{\circ }}{\partial\mathbf{x}^{\circ T}})^T,\ldots,
	(\frac{\partial r_{N_{a} 1}^{\circ }}{\partial\mathbf{x}^{\circ T}})^T,
	(\frac{\partial \dot{r}_{N_{a} 1}^{\circ }}{\partial\mathbf{x}^{\circ T}})^T,
	(\frac{\partial \phi_{1}^{\circ }}{\partial\mathbf{x}^{\circ T}})^T,
	(\frac{\partial \theta_{1}^{\circ }}{\partial\mathbf{x}^{\circ T}})^T,\ldots,
	(\frac{\partial \phi_{N_{a}}^{\circ }}{\partial\mathbf{x}^{\circ T}})^T,
	(\frac{\partial \theta_{N_{a}}^{\circ }}{\partial\mathbf{x}^{\circ T}})^T
	\!\right]^T\!\!\!\!,
	\vspace{-0.25cm}
	\end{equation}
and we have
\begin{equation}\label{35}
\hspace{-0.35cm}
\frac{\partial r_{i1}^{\circ }}{\partial\mathbf{x}^{\circ T}}\!=\!
\left[\frac{\partial r_{i1}^{\circ }}{\partial \mathbf{u}^{\circ T}},
\frac{\partial r_{i1}^{\circ }}{\partial \dot{\mathbf{u}}^{\circ T}}\right],
\frac{\partial \dot{r}_{i1}^{\circ }}{\partial\mathbf{x}^{\circ T}}\!=\!
\left[\frac{\partial \dot{r}_{i1}^{\circ }}{\partial \mathbf{u}^{\circ T}},
\frac{\partial \dot{r}_{i1}^{\circ }}{\partial \dot{\mathbf{u}}^{\circ T}}\right],
\frac{\partial \phi_{j}^{\circ }}{\partial\mathbf{x}^{\circ T}}\!=\!
\left[\frac{\partial \phi_{j}^{\circ }}{\partial \mathbf{u}^{\circ T}},
\frac{\partial \phi_{j}^{\circ }}{\partial \dot{\mathbf{u}}^{\circ T}}\right],
\frac{\partial \theta_{j}^{\circ }}{\partial\mathbf{x}^{\circ T}}\!=\!
\left[\frac{\partial \theta_{j}^{\circ }}{\partial \mathbf{u}^{\circ T}},
\frac{\partial \theta_{j}^{\circ }}{\partial \dot{\mathbf{u}}^{\circ T}}\right]\!\!,
\end{equation}
where $i=2,\ldots,N_{a}$ and $j=1,\ldots,N_{a}$.
According to the detailed derivations in Appendix \ref{D}, the elements of matrix $\mathbf{B}_1$ in (\ref{b1}) are given by
\vspace{-0.25cm}
\begin{equation}
\hspace{-6.5cm} \mathbf{B}_1(2i\!-\!3,:)\!=\!\left[(\mathbf{u}^{\circ}-\mathbf{b}_{i})^{T}/r_{i}^{\circ}-(\mathbf{u}^{\circ}-\mathbf{b}_{1})^{T}/r_{1}^{\circ}, \mathbf{0}\right],\nonumber
\end{equation}
\begin{equation}
\mathbf{B}_1(2i\!-\!2,:)\!=\!\bigg[\frac{\dot{r}_{1}^{\circ}(\mathbf{u}^{\circ}\!-\!\mathbf{b}_{1})^{T}}{(r_{1}^{\circ})^{2}}
\!-\frac{\dot{r}_{i}^{\circ}(\mathbf{u}^{\circ}\!-\!\mathbf{b}_{i})^{T}}{(r_{i}^{\circ})^{2}}\!+\!\frac{\dot{\mathbf{u}}^{\circ T}}{r_{i}^{\circ}}\!-\!\frac{\dot{\mathbf{u}}^{\circ T}}{r_{1}^{\circ}}, (\mathbf{u}^{\circ}-\mathbf{b}_{i})^{T}/r_{i}^{\circ}\!-\!(\mathbf{u}^{\circ}\!-\!\mathbf{b}_{1})^{T}/r_{1}^{\circ}\!\bigg],\nonumber
\end{equation}
\begin{equation}\label{99}
\hspace{-1.08cm}
\mathbf{B}_1(2N_{a}-3+2j,:)=\left[\mathbf{c}_{j}^{\circ T}/(r_{j}^{\circ}\cos\theta_{j}^{\circ}),\mathbf{0}\right],\ \ \ \ \mathbf{B}_1(2N_{a}-2+2j,:)=\left[\mathbf{d}_{j}^{\circ T}/r_{j}^{\circ},\mathbf{0}\right],
\vspace{-0.25cm}
\end{equation}
where $i=2,\ldots,N_{a}$ and $j=1,\ldots,N_{a}$.

A relevant question is whether the proposed estimator can achieve the CRLB under mild noise conditions.
For this reason, we will compare the covariance of our proposed estimator with the CRLB given in \eqref{b1}.
The covariance matrix of the estimator $\mathbf{x}$ is given by
\vspace{-0.15cm}
\begin{equation}\label{33}
\mbox{cov}(\mathbf{x})=\mathbb{E}\left\{(\mathbf{x}-\mathbb{E}\{\mathbf{x}\})(\mathbf{x}-\mathbb{E}\{\mathbf{x}\})^T\right\}.
\vspace{-0.25cm}
\end{equation}
Given $\mathbb{E}\{\mathbf{x}\} \doteq \mathbf{x}^{\circ}$,
we have
\vspace{-0.15cm}
\begin{equation}\label{313}
\mbox{cov}(\mathbf{x})\doteq \mathbb{E}\left\{(\mathbf{x}-\mathbf{x}^{\circ})(\mathbf{x}-\mathbf{x}^{\circ})^T\right\}
=\mathbb{E}\left\{\Delta\mathbf{x}\Delta\mathbf{x}^T\right\}.
\end{equation}
According to \eqref{32}, we have
\vspace{-0.15cm}
\begin{equation}\label{332}
\begin{aligned}
\mbox{cov}(\mathbf{x})
&\doteq \mathbb{E}\left\{(\tilde{\mathbf{G}}^{T}\mathbf{W}\tilde{\mathbf{G}})^{-1}\tilde{\mathbf{G}}^{T}\mathbf{W}(\mathbf{e}\mathbf{e}^T)\mathbf{W}^T\tilde{\mathbf{G}}(\tilde{\mathbf{G}}^{T}\mathbf{W}^{T}\tilde{\mathbf{G}})^{-1} \right\}\\
&\doteq (\mathbf{G}^{T}\mathbf{W}\mathbf{G})^{-1}\mathbf{G}^{T}\mathbf{W}\mathbb{E}\!\left\{\mathbf{e}\mathbf{e}^T\!\right\}\mathbf{W}^T\mathbf{G}(\mathbf{G}^{T}\mathbf{W}^{T}\mathbf{G})^{-1}.
\end{aligned}
\vspace{-0.25cm}
\end{equation}
In addition, given $\mathbf{W}^{-1} = \mathbb{E}\!\left\{\mathbf{e}\mathbf{e}^T\!\right\}$ (see Appendix \ref{B}), we can obtain
\vspace{-0.25cm}
\begin{equation}\label{333}
\mbox{cov}(\mathbf{x})
\doteq(\mathbf{G}^{T}\mathbf{W}\mathbf{G})^{-1}.
\vspace{-0.25cm}
\end{equation}
Substituting $\mathbf{W}$ in (\ref{31}) into (\ref{333}) implies
\vspace{-0.25cm}
\begin{equation}
\mbox{cov}(\mathbf{x})\doteq (\mathbf{G}^{T}(\mathbf{B}\mathbf{Q}\mathbf{B}^T)^{-1}\mathbf{G})^{-1}=\left((\mathbf{B}^{-1}\mathbf{G})^T\mathbf{Q}^{-1}(\mathbf{B}^{-1}\mathbf{G})\right)^{-1}.
\vspace{-0.25cm}
\end{equation}
By denoting $\mathbf{B}_0=\mathbf{B}^{-1}\mathbf{G}$, we have
\vspace{-0.25cm}
\begin{equation}\label{34}
\mbox{cov}(\mathbf{x})\doteq (\mathbf{B}_0^T\mathbf{Q}^{-1}\mathbf{B}_0)^{-1}.
\vspace{-0.25cm}
\end{equation}
By comparing the CRLB expressed in \eqref{b1} with the covariance of the proposed estimator given in \eqref{34}, we can derive the following theorem:
\vspace{-0.25cm}
\begin{Theorem}\label{theorem2}
	The proposed estimator in (\ref{331}) attains the lowest possible variance of CRLB under low Gaussian noise, i.e., $\mbox{\rm cov} (\mathbf{x})\doteq\mbox{\rm CRLB}(\mathbf{x}^{\circ})$, when the weighting matrix in (\ref{31}) is used.
\end{Theorem}
\begin{proof}
See Appendix \ref{E}.	
\end{proof}
\vspace{-0.25cm}
\begin{remark}
Uplink localization in the mmWave can be executed in the initial access stage, as explained in Section \ref{ra}.
Uplink localization can also be collaterally executed in the uplink channel estimation stage in which the estimated channel parameters are used to obtain the position and velocity information of the UE. Then, the location information can be used to predict the state of the UE at the next epoch to facilitate the initial access, beam training, and channel tracking processes.
\end{remark}

\vspace{-0.5cm}
\section{Environment Mapping}
In this section, we present an effective environment mapping method.
As mentioned previously, for ease of representation, we assume that
(1) an unknown scatterer exists between the $n$-th RRH and the UE, where $n=1,2,\ldots,N$,
and (2) the scatterer contributes to a single-bounce NLoS path.
For the $n$-th RRH, we define a vector of measured parameters $\mathbf{m}_n^{s}=[r_{n1}^{s}, \phi_n^{s}, \theta_n^{s}]^T$ that correspond to the $n$-th scatterer $\mathbf{s}_n^\circ$, where $r_{n1}^{s}$ is the measurement of the TDoA-related parameter defined in \eqref{sc1}, and $\phi_n^{s}$ and $\theta_n^{s}$ are the measurements of the AoA-related parameters defined in \eqref{55}.
Similar to \eqref{m0}, we define $\mathbf{m}_n^{s\circ}=[r_{n1}^{s\circ}, \phi_n^{s\circ}, \theta_n^{s\circ}]^T$ as the vector of the true parameters.
Then, we have
\vspace{-0.25cm}
\begin{equation}
\mathbf{m}_n^{s} = \mathbf{m}_n^{s\circ} + \Delta\mathbf{m}_n^{s},
\vspace{-0.25cm}
\end{equation}
where $\Delta\mathbf{m}_n^{s}=[\Delta r_{n1}^{s}, \Delta\phi_n^{s}, \Delta\theta_n^{s}]^T$ denotes the vector of the noise terms.
We assume that $\Delta\mathbf{m}_n^{s}$ is a zero-mean Gaussian vector with covariance matrix $\mathbf{Q}^s_n$.
According to \eqref{sc1},  we have
\vspace{-0.25cm}
\begin{equation} \label{sc11}
r_{n1}^{s\circ} \! = \! ||\mathbf{u}^\circ\!-\!\mathbf{s}_n^\circ|| \! + \! ||\mathbf{s}_n^\circ\!-\!\mathbf{b}_n||\! -\! ||\mathbf{u}^\circ\!-\!\mathbf{b}_1||.
\end{equation}
Let $d_{n1}^\circ = ||\mathbf{s}_n^\circ-\mathbf{b}_n||$ and
$d_{n2}^\circ = ||\mathbf{u}^\circ-\mathbf{s}_n^\circ||$. Then,
we have
\begin{equation} \label{s1}
r_{n1}^{s\circ} = d_{n1}^\circ + d_{n2}^\circ - r_1^\circ.
\end{equation}
By denoting $J^\circ=r_{n1}^{s\circ} + r_1^\circ$, we have $d_{n2}^\circ = J^\circ-d_{n1}^\circ$.
By squaring both sides, we have
\begin{equation} \label{s4}
(J^\circ)^2 + (d_{n1}^\circ)^2 - 2 J^\circ d_{n1}^\circ =(d_{n2}^\circ)^2.
\vspace{-0.25cm}
\end{equation}
We simplify \eqref{s4} to obtain
\vspace{-0.25cm}
\begin{equation} \label{s5}
\begin{split}
(J^\circ)^2 \!-\! 2 J^\circ d_{n1}^\circ
= & {\mathbf{u}^\circ}^T\mathbf{u}^\circ-2{\mathbf{u}^\circ}^T\mathbf{s}_n^\circ + 2 {\mathbf{b}_n}^T\mathbf{s}_n^\circ - {\mathbf{b}_n}^T\mathbf{b}_n.
\end{split}
\vspace{-0.25cm}
\end{equation}
We also denote
\begin{equation} \label{s6}
\begin{split}
&\mathbf{a}_n^{s\circ} = [\cos\theta_n^{s\circ}\cos\phi_n^{s\circ}, \cos\theta_n^{s\circ}\sin\phi_n^{s\circ}, \sin\theta_n^{s\circ}]^T,\ \ \ \ \ \ \ \ \mathbf{c}_n^{s\circ}=[-\sin\phi_n^{s\circ}, \cos\phi_n^{s\circ}, 0]^T,\\
&\mathbf{d}_n^{s\circ}=[-\sin\theta_n^{s\circ}\cos\phi_n^{s\circ}, -\sin\theta_n^{s\circ}\sin\phi_n^{s\circ}, \cos\theta_n^{s\circ}]^T,
\end{split}
\end{equation}
where ${\mathbf{a}_n^{s\circ}}^T\mathbf{a}_n^{s\circ}=1$
and $\mathbf{s}_n^\circ-\mathbf{b}_n=d_{n1}^\circ\mathbf{a}_n^{s\circ}$.
Similar to those in Appendix \ref{A}, we derive
\vspace{-0.25cm}
\begin{equation} \label{s10}
\begin{split}
(J^\circ)^2 + 2 J^\circ{\mathbf{a}_n^{s\circ}}^T\mathbf{b}_n - {\mathbf{u}^\circ}^T\mathbf{u}^\circ+ {\mathbf{b}_n^T}\mathbf{b}_n &= 2({\mathbf{b}_n}-{\mathbf{u}^\circ}+J^\circ {\mathbf{a}_n^{s\circ}})^T\mathbf{s}_n^\circ,
\end{split}
\vspace{-0.25cm}
\end{equation}
and
\vspace{-0.25cm}
\begin{equation} \label{s11}
{\mathbf{c}_n^{s\circ}}^T\mathbf{b}_n ={\mathbf{c}_n^{s\circ}}^T\mathbf{s}_n^\circ, \ \ \
{\mathbf{d}_n^{s\circ}}^T\mathbf{b}_n ={\mathbf{d}_n^{s\circ}}^T\mathbf{s}_n^\circ.
\vspace{-0.25cm}
\end{equation}
By stacking \eqref{s10} and \eqref{s11},
we derive the environment mapping model, which is given by
\vspace{-0.25cm}
\begin{equation} \label{s12}
\mathbf{h}^{s}_{n}=\mathbf{G}^{s}_{n}\mathbf{s}_n^\circ,
\vspace{-0.25cm}
\end{equation}
where
\begin{eqnarray} \label{s13}
\mathbf{h}^{s}_{n} = \begin{pmatrix}
(J^\circ)^2 + 2 J^\circ{\mathbf{a}_n^{s\circ}}^T\mathbf{b}_n - {\mathbf{u}^\circ}^T\mathbf{u}^\circ + {\mathbf{b}_n}^T\!\mathbf{b}_n \\
{\mathbf{c}_n^{s\circ}}^T\mathbf{b}_n\\
{\mathbf{d}_n^{s\circ}}^T\mathbf{b}_n
\end{pmatrix}, \ \ \mathbf{G}^{s}_{n} = \begin{pmatrix}
2({\mathbf{b}_n}-{\mathbf{u}^\circ}+J^\circ {\mathbf{a}_n^{s\circ}})^T\\
{\mathbf{c}_n^{s\circ}}^T\\
{\mathbf{d}_n^{s\circ}}^T
\end{pmatrix}.
\end{eqnarray}
Let the measurements $\{r_{n1}^{s}, \phi_n^{s},  \theta_n^{s}\}$ replace the true parameters $\{r_{n1}^{s\circ}, \phi_n^{s\circ}, \theta_n^{s\circ}\}$ in \eqref{s12}, then, we have
\vspace{-0.1cm}
\begin{equation} \label{s14}
\tilde{\mathbf{h}}^{s}_{n}=\tilde{\mathbf{G}}^{s}_{n}\mathbf{s}_n^\circ + \mathbf{e}_n^s,
\vspace{-0.2cm}
\end{equation}
where $\mathbf{e}_n^s$ is the error vector caused by the measurement error.

Subtracting \eqref{s12} from \eqref{s14} obtains $\mathbf{e}_n^s$.
By approximating $\mathbf{e}_n^s$ up to the linear noise term, we have
$
\mathbf{e}_n^s=[e_1, e_2, e_3]^T \doteq \mathbf{B}_n^s\Delta \mathbf{m}_n^s,
$
where
\vspace{-0.1cm}
\begin{equation}\label{s16}
\mathbf{B}_n^s=\mbox{diag}(
2d_{n2}^\circ,
d_{n1}^\circ \cos\theta_n^{s\circ} ,
d_{n1}^\circ).
\vspace{-0.2cm}
\end{equation}
The derivations of \eqref{s16} are similar to those in Appendix \ref{B},
but we omit in this paper the details due to lack of space.
Similarly, according to Lemma 1, the environment mapping estimator is given by
\begin{equation}\label{s17}
\mathbf{s}_n=
({\mathbf{\tilde{G}}^{sT}_n}\mathbf{W}^s_n{\mathbf{\tilde{G}}^s_n})^{-1}
{\mathbf{\tilde{G}}^{sT}_n}
\mathbf{W}^s_n
\mathbf{\tilde{h}}^s_n
\vspace{-0.25cm}
\end{equation}
with
\vspace{-0.15cm}
\begin{equation}\label{s0}
\mathbf{W}^s_n=(\mathbf{B}^s_n\mathbf{Q}^s_n{\mathbf{B}^s_n}^T)^{-1}.
\vspace{-0.25cm}
\end{equation}
We summarize the proposed environment mapping algorithm in Algorithm \ref{alg2}.
\begin{algorithm}
	\caption{\textbf{: Environment Mapping}}\label{alg2}
	\begin{algorithmic}[1]
		\Require $\mathbf{Q}^s_n$, $\mathbf{b}_n$, and measurements set $\mathbb{S}_n^s=\{r_{n1}^{s}, \phi_n^{s},  \theta_n^{s}\}$, for $n=1,\ldots,N$.
		\Ensure $\mathbf{s}_n$
		\State Initialization: $t=1$, $\mathbf{W}^s_n={\mathbf{Q}^s_n}^{-1}$.
		\State Calculate $\tilde{\mathbf{h}}^s_n$ and $\tilde{\mathbf{G}}^s_n$ by the measurements in $\mathbb{S}_n^s$ and the UE position $\mathbf{u}^\circ$ obtained in Algorithm \ref{alg1}.
		\State Calculate the initial $\mathbf{s}_n$ in \eqref{s17} with $\mathbf{W}^s_n={\mathbf{Q}^s_n}^{-1}$.
		\While {$t\leq T$}
		\State Calculate $d_{n1}^\circ$ and $d_{n2}^\circ$ by the obtained $\mathbf{s}_n$.
		\State Generate $\mathbf{B}^s_n$ by parameters obtained in step $5$.
		\State Update weighting matrix $\mathbf{W}^s_n$ in \eqref{s0} by $\mathbf{B}^s_n$.
		\State Update $\mathbf{s}_n$ according to \eqref{s17}.
		\State $t=t+1$.
		\EndWhile
		\State \Return{$\mathbf{s}_n$}
	\end{algorithmic}
\end{algorithm}

\vspace{-0.5cm}
\section{Neural Network-Assisted WLS Algorithm}
The WLS-based joint estimation and mapping algorithms proposed in Sections \uppercase\expandafter{\romannumeral3} and \uppercase\expandafter{\romannumeral4} are proven asymptotically unbiased and effective in achieving the CRLB under mild noise conditions.
The general assumption is that the vector of the noise terms follow a Gaussian distribution, and the covariance matrix is $\mathbf{Q}$ (or $\mathbf{Q}_n^s$).
This assumption means that the proposed WLS-based joint estimation and mapping algorithms have superior performance and wide versatility.

However, the performance of the proposed WLS algorithms can be further improved under some practical conditions.
In particular, this study is motivated by the observation of the ray-tracing dataset provided in
\cite{deepmimo,raytracing} that the measurement errors are not completely random and that an underlying relationship exists between them.
Hence, by utilizing the powerful learning ability of the neural networks, the underlying relationship can be learned to further improve the localization performance of the proposed WLS algorithms.

In this section,
we will design a WLS-Net that embeds the neural networks into the proposed WLS estimators in \eqref{21} and \eqref{s17},
thus improving the performance of the joint estimation and mapping algorithms. Different from the traditional neural network methods that can directly learn position and velocity,
the neural network in our approach is used to learn the residual vectors $\mathbf{e}$ in \eqref{20} and $\mathbf{e}_n^s$ in \eqref{s14} instead of learning $\mathbf{x}^{\circ}$ and $\mathbf{s}_n^\circ$, respectively.
Then, the estimated $\mathbf{e}$ and $\mathbf{e}_n^s$ are used to construct the weighting matrices $\mathbf{W}$ and $\mathbf{W}_n^s$ in \eqref{21} and \eqref{s17} and estimate $\mathbf{x}^{\circ}$ and $\mathbf{s}_n^\circ$, respectively.
By learning the residual vectors, we can derive more accurate results than directly learning position and velocity.
We will also apply ensemble learning to further improve the performance of the proposed method.
\begin{figure*}[htbp]
	\centering
	\includegraphics[scale=0.63,angle=0]{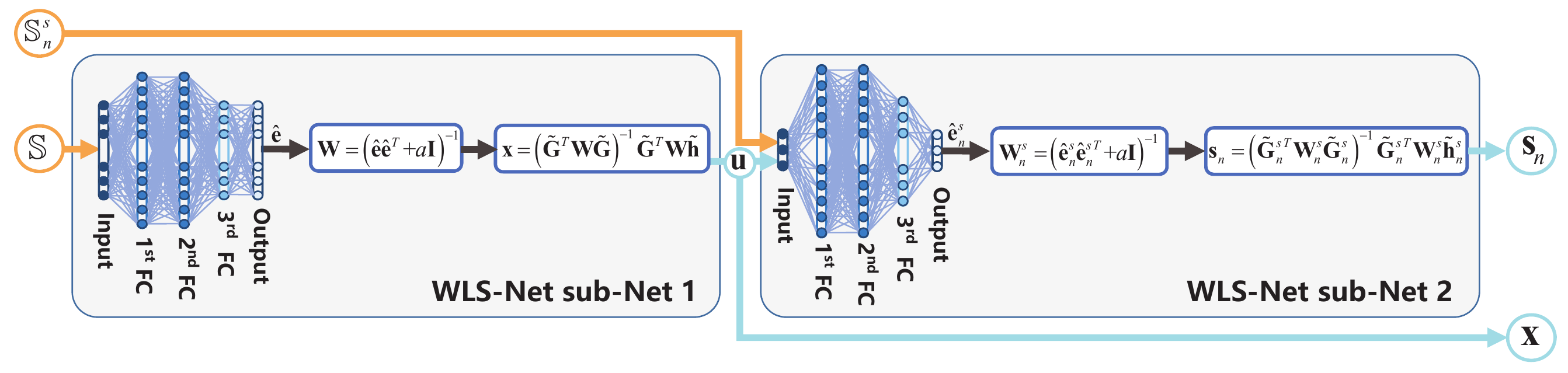}
	\captionsetup{font=footnotesize}
	\caption{Structure of the neural network-assisted WLS algorithm with  two sub-nets.}
	\label{WLSNET}
\end{figure*}

\vspace{-0.25cm}
\subsection{WLS-Net}\label{wlsnet}
The structure of the proposed WLS-Net is illustrated in Fig. \ref{WLSNET}, which is a revised version of the WLS algorithm derived by adding the learnable vectors $\mathbf{e}$ and $\mathbf{e}_n^s$.
Sub-Nets 1 and 2 of the WLS-Net have similar structures. The former estimates the position and velocity of the UE, whereas the latter estimates the position of the scatterer.
We give a general introduction here by taking sub-Net 1 of the WLS-Net as an example.
According to \cite{KAY}, the weighting matrix that minimizes the Frobenius norm of the covariance matrix of $\mathbf{x}$ is given by $\mathbf{W}=(\mathbb{E}\{\mathbf{e}\mathbf{e}^{T}\})^{-1}$.
In the WLS algorithm proposed in the previous sections, the vector $\mathbf{e}$ is approximated by the linear term. Hence, as the noise level increases, the approximation error will become larger, which will deteriorate algorithm performance.
Therefore, we propose the WLS-Net in which we learn the vector $\mathbf{e}$ by a neural network.
The input of the neural network is a set of measurements  $\mathbb{S}$ and the output of the neural network is the estimated residual vector $\hat{\mathbf{e}}$.
The proposed WLS-Net differs from the traditional neural network methods in that the latter directly outputs the estimate of $\mathbf{x}$.
Then, the estimated $\hat{\mathbf{e}}$ in the WLS-Net is used to construct $\mathbf{W}$ by
\begin{equation}
 \mathbf{W}=(\hat{\mathbf{e}}\hat{\mathbf{e}}^{T}+a\mathbf{I})^{-1},
\end{equation}
  where $a$ is a very small disturbance value to ensure the inverse of $(\hat{\mathbf{e}}\hat{\mathbf{e}}^{T}+a\mathbf{I})$ will exist.
Finally, we obtain the estimate by the model $\mathbf{x}=(\tilde{\mathbf{G}}^{T}\mathbf{W}\tilde{\mathbf{G}})^{-1}\tilde{\mathbf{G}}^{T}\mathbf{W}\tilde{\mathbf{h}}$.

As shown in Fig. \ref{WLSNET}, the input of the WLS-Net sub-Net 1 is a set of measurements given by
\begin{equation}\label{set1}
\mathbb{S}=\{r_{i1},\dot{r}_{i1},\phi_j,\theta_j\},
\end{equation}
where $i=2,\ldots,N_{a}$ and $j=1,\ldots,N_{a}$.
For different $N_{a}$, the neural network will have varying numbers of neurons. Here, we set $N_{a}=6$ as an example.
Sub-Net 1 consists of a three-layer fully connected deep neural network (FC-DNN).
As shown in Fig. \ref{WLSNET},
the input of the neural network is a 22-dimensional ($4N_{a}-2=22$) real-valued normalized measurements vector with the value of the element in $[0,1]$.
The first-two FC layers use 32 neurons, and the third FC layer uses 22 neurons.
As for the rectified linear unit (ReLU), $\mbox{ReLU}(x)=\max(x,0)$ is used as the activation function.
The final layer is the output layer used to generate the final estimation of $\mathbf{e}$, which is denoted as $\hat{\mathbf{e}}$. The sigmoid function
$\sigma(x)=1/(1+e^{-x})$ is used as the activation function in the final layer because the output is the normalized vector which has elements scaled within the $[0,1]$ range.
The set of parameters is updated by the ADAM algorithm.
The loss function is the mean square error (MSE), which is given by
\begin{equation}\label{loss}
L(\Theta)=1/T_a\sum_{i=1}^{T_a}\lVert \hat{\mathbf{e}}_i - \mathbf{e}_i \rVert^2,
\end{equation}
where $T_a$ is the total number of samples in the training set.

Similarly, the input of the WLS-Net sub-Net 2 is a set of measurements that are given by
$\mathbb{S}_n^s=\{r_{n1}^{s}, \phi_n^{s}, \theta_n^{s}\}$,
where $n=1,\ldots,N$.
The location vector $\mathbf{u}$ of the UE obtained in the WLS-Net sub-Net 1 is also fed into sub-Net 2, which is viewed as a known vector that can assist the construction of the environment mapping model.
Sub-Net 2 has the same FC-DNN architecture as sub-Net 1, except that the input and the output layers have three neurons.
The output is the estimation of $\mathbf{e}_n^s$ and denoted as $\hat{\mathbf{e}}_n^s$.
The WLS-Net sub-Net 2 will be executed $N$ times to compute the location of the $N$ scatterers.
\begin{remark}
	The proposed WLS-Net combines the neural networks with the geometric model, thus inheriting the powerful computing ability of neural networks and the robustness of models. The advantages in particular are as follows. First, the neural networks can provide a more accurate estimation of $\mathbf{e}$ and $\mathbf{e}_n^s$ than the first-order approximation in the previously proposed WLS algorithms. Hence, in some practical scenarios, the WLS-Net will achieve good performance. Moreover, the WLS-Net can be executed even without knowing the covariance matrix $\mathbf{Q}$ and $\mathbf{Q}_n^s$, whereas the $\mathbf{Q}$ and $\mathbf{Q}_n^s$ in the WLS algorithms are assumed to be known to initialize the weighting matrix $\mathbf{W}$ and $\mathbf{W}^s_n$, respectively.
	Finally, the WLS algorithm requires iterations, which implies slow reconstruction, whereas the WLS-Net does not need any iterations, thereby reducing the required time resources.
	The WLS algorithm also has its own irreplaceable advantages in that large amounts of training data are not needed and the versatility can be further enhanced. Therefore, each of these two methods can be selected depending on actual computational resources and buffer restrictions.
\end{remark}

\vspace{-0.25cm}
\subsection{Ensemble Learning-Based WLS-Net}\label{ewlsnet}
Training the WLS-Net with the loss function defined in \eqref{loss} despite the sufficient data still cannot guarantee that the WLS-Net will output the globally optimal estimator.
According to \cite{ensemble}, the ensemble learning methods often outperform a single learner.
Ensemble methods are learning algorithms that construct a set of learners and generate a new prediction by taking a vote, which may be weighted, of the predictions.
In the backpropagation algorithm for training the neural networks, the initial weights of the networks are set randomly.
If the algorithm is applied to the same training dataset but with different initial weights, then the resulting predictions may differ.
Neural networks independently trained with the same training dataset have high probabilities of not making the same prediction error. Therefore, we can further improve the performance of the neural network-assisted WLS algorithm by proposing an ensemble learning-based WLS-Net, also named eWLS-Net, which is an ensemble of $L$-independently trained WLS-Nets.

\vspace{-0.5cm}
\begin{figure}[htbp]
	\centering
	\includegraphics[scale=0.575,angle=0]{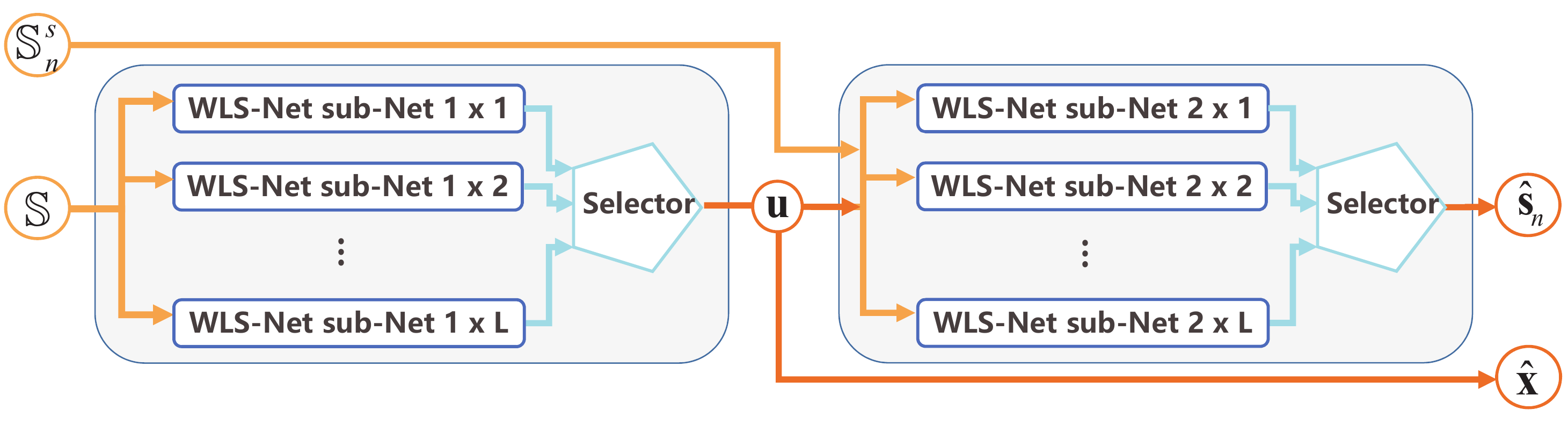}
	\captionsetup{font=footnotesize}
	\caption{Structure of the ensemble learning-based WLS-Net.}
	\label{eWLSNET}
\end{figure}
\vspace{0.3cm}
The structure of the eWLS-Net is illustrated in Fig. \ref{eWLSNET}.
Take the ensemble of WLS-Net sub-Net 1 as an example. The ensemble consists of $L$-independently trained WLS-Net sub-Net 1 and a selector.
The input of the ensemble is a set of measurements given in \eqref{set1}, and the output is the estimated $\hat{\mathbf{x}}$.
The core part of the eWLS-Net is the selector, which determines the voting mechanism.
We need to choose different selectors \cite{ensemble,enet} to solve different practical problems.
Here, the output of each WLS-Net sub-Net 1 is an independent prediction of $\mathbf{x}$, where $\mathbf{x}=(\mathbf{u}^{T},\dot{\mathbf{u}}^{T})^T$.
Hence, we obtain $L$ independent predictions of 3-D position coordinates $\mathbf{u}$ and velocity coordinates $\dot{\mathbf{u}}$.
We place $L$ predictions of $\mathbf{u}$ in the 3-D coordinate system.
The accurate predictions of $\mathbf{u}$ are clustered together, and the wrong predictions are positioned far apart, which is the same approach for $\dot{\mathbf{u}}$.
Therefore, we implement the selector in this study by the subtractive clustering method, which can find the clustering centers based on a density measure.
Subtractive clustering is operated separately for $\mathbf{u}$ and $\dot{\mathbf{u}}$.
Take $\mathbf{u}$ for example. The $L$ predictions of $\mathbf{u}$ are viewed as $L$ points and considered the candidates for cluster centers.
The density measurement obtained at point $\mathbf{u}_i$ for $i=1,\ldots,L$ is defined as
\begin{equation}\label{den}
D_i=\sum\limits_{j=1}^{L}\exp\left(- \lVert \mathbf{u}_i - \mathbf{u}_j \rVert^2 /(r_a/2)^2\right),
\end{equation}
where $r_a$ is a positive value to denote the radius. The data points outside this radius contribute only slightly to the density measurement.
After the density measurement of each data point has been calculated, the data point
with the highest density measurement $D_{c1}$ is selected as the first cluster center $\mathbf{u}_{c1}$.
Then, the density measurement of each data point is corrected by
\begin{equation}\label{den}
D_i^{'}=D_i-D_{c1}\exp\left(- \lVert \mathbf{u}_i - \mathbf{u}_{c1} \rVert^2 /(r_b/2)^2\right),
\end{equation}
where $r_b$ is a positive value and normally set to be larger than $r_a$ to prevent the occurrence of closely spaced cluster centers.
After the density measurement of each data point is revised, the next cluster center $\mathbf{u}_{c2}$ is selected,
and all density measurements of the data points are revised again.
This process is repeated until a sufficient number of cluster centers are generated.
We can infer that the first cluster center $\mathbf{u}_{c1}$ has the largest density among all centers based on the above process, and $\mathbf{u}_{c1}$ is thereby selected as the final estimation of $\mathbf{u}$.
Similarly, we denote $\dot{\mathbf{u}}_{c1}$ as the first cluster center of $L$ predictions of $\dot{\mathbf{u}}$.
Finally,
the output of the selector in the eWLS-Net sub-Net 1 is $\hat{\mathbf{x}} = (\mathbf{u}_{c1}^T,\dot{\mathbf{u}}_{c1}^T)^T$.

\vspace{-0.25cm}
\section{Numerical Results}
\subsection{WLS-Based Joint Position and Velocity Estimation}\label{sp}
In this section, we analyze the performance of the proposed WLS algorithm in jointly estimating the position and velocity of the UE.
We consider a scenario with $N=6$ RRHs, located at $[-400,0,0]^T$, $[400,0,0]^T$, $[200,350,0]^T$, $[-200,350,0]^T$, $[-200,-350,0]^T$, and $[200,-350,0]^T$ in meters, respectively.
The UE is located at $\mathbf{u}^\circ=[300,-20,-100]^T$ (m) with the velocity $\dot{\mathbf{u}}^\circ=[-9,7,5]^T$ (m/s).
The CU selects $N_{a}=2$ to $6$ RRHs with LoS paths to locate the UE with the proposed localization method.

The covariance matrix of the noise terms defined in \eqref{m} is given by
{\begingroup\makeatletter\def\f@size{9}\check@mathfonts		 \def\maketag@@@#1{\hbox{\m@th\normalsize\normalfont#1}}\setlength{\arraycolsep}{0.0em}\setlength{\arraycolsep}{0.0em}
	\begin{eqnarray}\label{mm}
	\mathbf{Q}\!\!=\!\mbox{blkdiag}(\overbrace{\mathbf{Q}_d,\ldots,\mathbf{Q}_d}^{(N_{a}-1)},
	\overbrace{\mathbf{Q}_a,\ldots,\mathbf{Q}_a}^{N_{a}}),
	\end{eqnarray}\setlength{\arraycolsep}{5pt}\endgroup}where  $\mathbf{Q}_d=\mbox{diag}(\sigma_d^2,(0.1\sigma_d)^2)$,
$\mathbf{Q}_a=\mbox{diag}(\sigma_a^2,\sigma_a^2)$, and $\sigma_d$, $0.1\sigma_d$, and $\sigma_a$ are the standard deviations of TDoA, FDoA, and AoA measurements.
The order of the elements in \eqref{mm} is the same as that in \eqref{m}, in which the first $\!(\!N_{a}\!-\!1\!)\!$ pairs are TDoA and FDoA pairs (the covariance matrix for each pair is $\mathbf{Q}_d$), and the last $N_{a}$ pairs are AoA pairs (i.e., the covariance matrix for each pair is $\mathbf{Q}_a$).
The total number of iterations $T$ is set to $5$ (the algorithm has converged).
The localization accuracy is
assessed via the root mean square error $\mbox{RMSE}(\mathbf{u})=\sqrt{\sum_{i=1}^{L_a}||\mathbf{u}_i-\mathbf{u}^\circ||^2/L_a}$ and $\mbox{RMSE}(\dot{\mathbf{u}})=\sqrt{\sum_{i=1}^{L_a}||\dot{\mathbf{u}}_i-\dot{\mathbf{u}}^\circ||^2/L_a}$,
where $\mathbf{u}_i$ and $\dot{\mathbf{u}}_i$ are
the estimation of $\mathbf{u}^\circ$ and $\dot{\mathbf{u}}^\circ$ at the $i$-th Monte Carlo simulation, respectively.
All the numerical results provided in this section are obtained from $L_a=1000$ independent Monte Carlo simulations.

\begin{figure}[t]
	\centering
	\hspace{-0.25cm}
	\begin{minipage}[t]{0.48\textwidth}
		\centering
		\includegraphics[scale=0.46,angle=0]{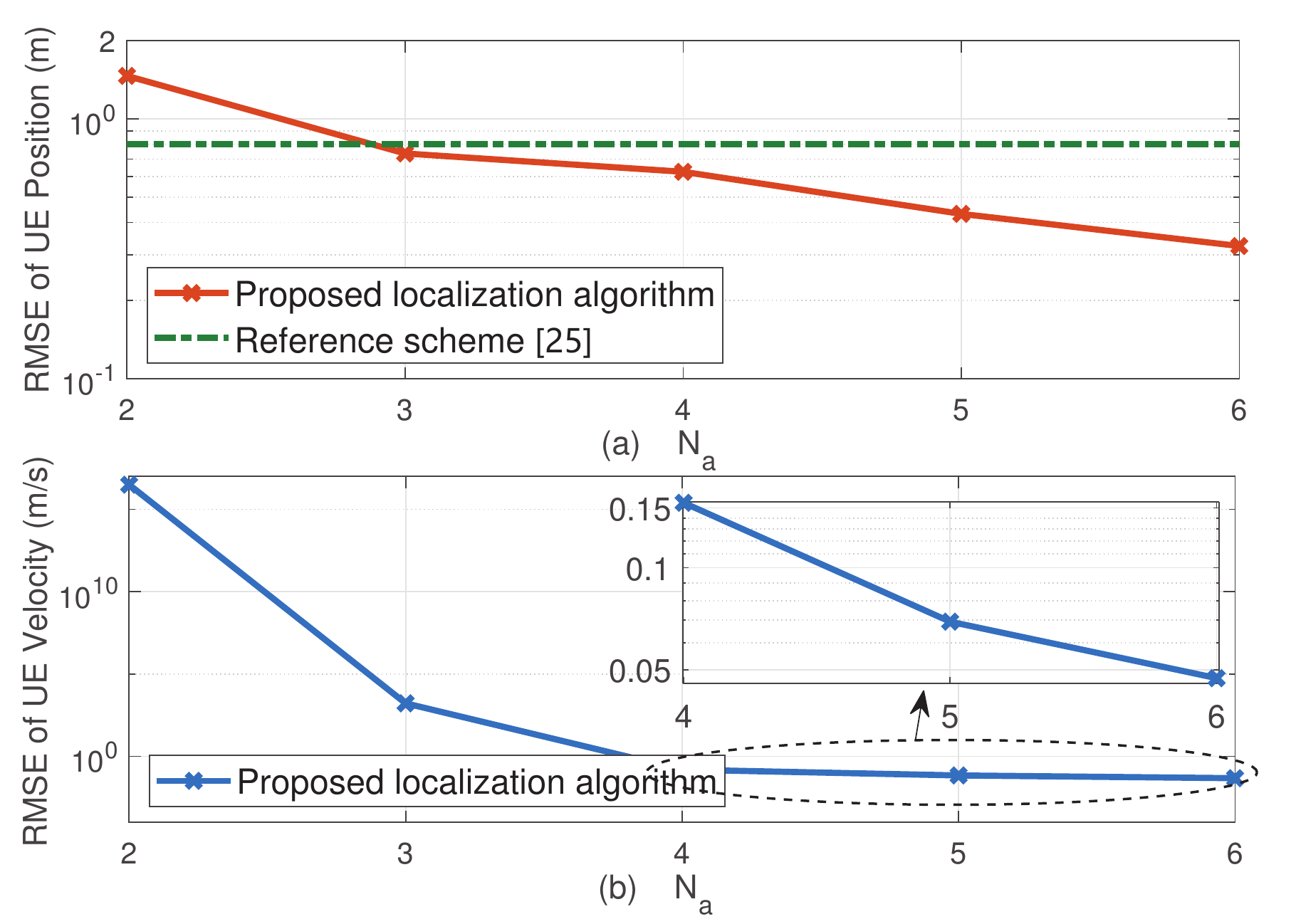}
		\captionsetup{font=footnotesize}
		\caption{$\mbox{RMSE}$ as a function of $N_{a}$ for (a) position of UE and (b) velocity of UE.}\label{fig:rrh}
	\end{minipage}
	\hspace{0.25cm}
	\begin{minipage}[t]{0.48\textwidth}
		\centering
		\includegraphics[scale=0.43,angle=0]{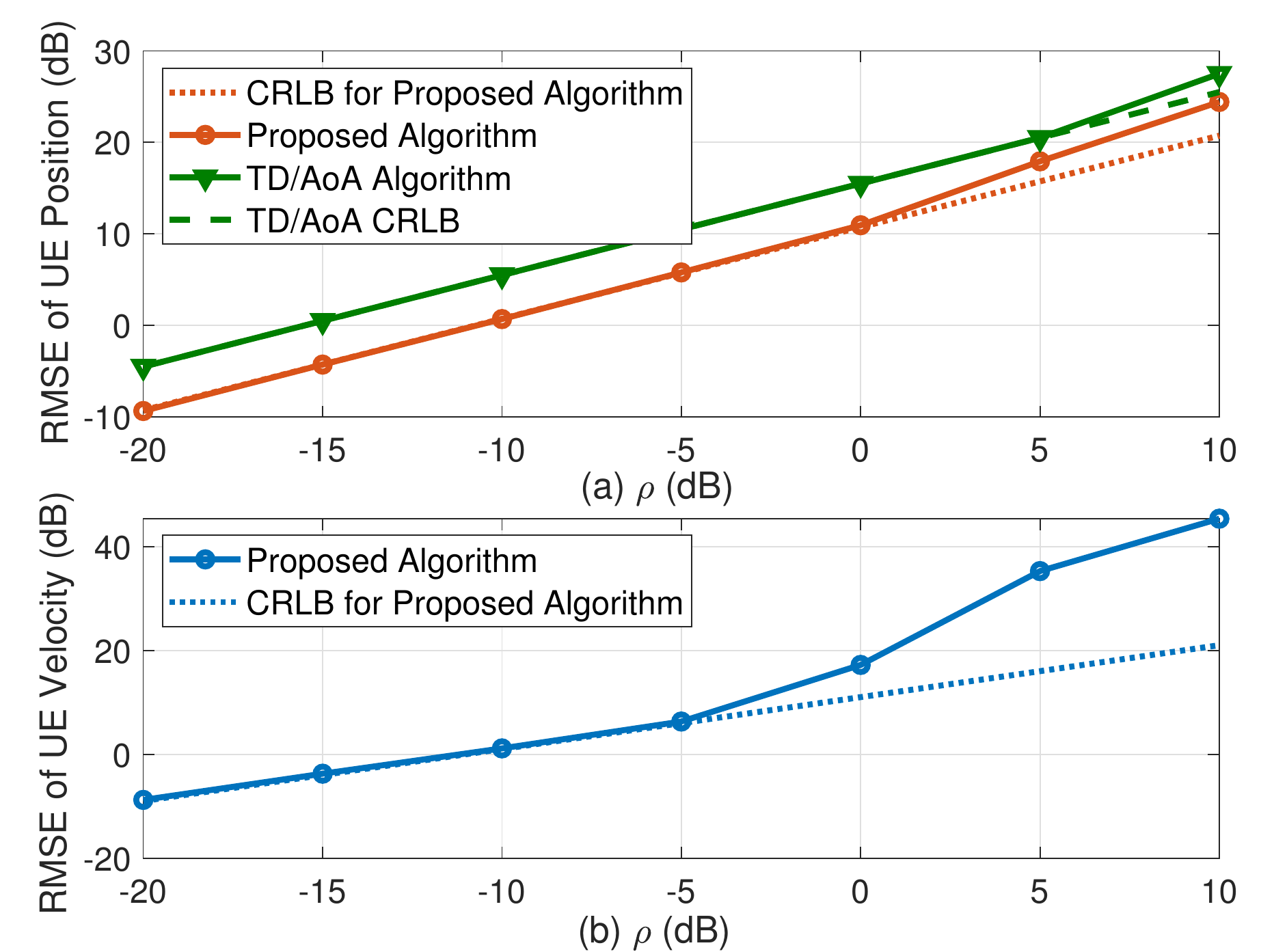}
		\captionsetup{font=footnotesize}
		\caption{Performance analysis of the presented positioning algorithm by comparison with the TD/AoA algorithm and the CRLB as a function of noise scaling factor for (a) position of UE and (b) velocity of UE.}\label{fig:crlb}
	\end{minipage}
\end{figure}

In the first simulation scenario, we analyze the performance of the proposed algorithm with different $N_{a}$ by setting $\sigma_d = 0.1$ (m) and $\sigma_a = 0.01$ (rad).
The reference scheme in Fig. \ref{fig:rrh}(a) shows the performance of the downlink positioning proposed in \cite{l5} with uninformative clock bias and UE orientation, which is achieved by combining more than $4$ single-bounce NLoS paths.
We observe a number of interesting facts from Fig. \ref{fig:rrh}.
In all cases, having a larger $N_{a}$ is beneficial, which can be explained by Theorem \ref{theorem1}.
We can achieve nearly the same accuracy with the reference scheme by simply using  $2$ RRHs,
while using $3$ or more RRHs in our scheme can outperform the latter.
Compared with the downlink localization, the proposed uplink localization will not be affected by the unknown antenna array orientation of the UE.
To our best knowledge, no other TDoA/FDoA/AoA-based joint velocity and position estimation algorithm in 3-D space can be compared with our proposed method in Fig. \ref{fig:rrh}(b).
Nevertheless, we can realize sub-decimeter per-second-level accuracy for velocity estimation with $4$ or more RRHs, as shown in Fig.  \ref{fig:rrh}(b).

In the second scenario, we consider the CRLB as the benchmark of accuracy of the proposed algorithm.
We set $N_{a} = 6$, $\sigma_d = 40 \rho$, and $\sigma_a = 0.1\rho$, where $\rho$ is a scaling factor.
The results shown in Fig. \ref{fig:crlb} verify that the proposed TDoA/AoA positioning is more accurate than the TD/AoA positioning presented in \cite{l2} due to the cancelation of the clock bias, and it can achieve the CRLB for small $\rho$ (noise level).
Increasing the noise level results in a slow deviation from the CRLB for both position and velocity estimations because the nonlinear terms in $\mathbf{e}$ in the derivation of our algorithm have been ignored.

In the third scenario, we analyze the performance of the proposed joint estimation algorithm with different $\sigma_d$ and $\sigma_a$ by setting $N_{a}=6$.
Fig. \ref{fig:sigma} demonstrates that the proposed algorithm can further enhance precision by using highly accurate channel parameter measurements.
Position is mainly determined by the AoA and TDoA measurements;
hence, the RMSE of UE position is affected by $\sigma_d$ and $\sigma_a$ at the same time.
However, velocity is mainly determined by the FDoA measurements; hence, the RMSE of the UE velocity is mainly affected by $\sigma_d$ when the value of $\sigma_d$ is large.
In particular, with $\sigma_d < -10$ dB and $\sigma_a \leq 0.1^\circ$, the estimated position can achieve sub-decimeter level accuracy.
Moreover, with $\sigma_d < 0$ dB and $\sigma_a \leq 1^\circ$, the estimated velocity can achieve sub-decimeter per-second-level accuracy.
\begin{figure}[t]
	\centering
	\hspace{-0.25cm}
	\begin{minipage}[t]{0.48\textwidth}
		\centering
		\includegraphics[scale=0.45,angle=0]{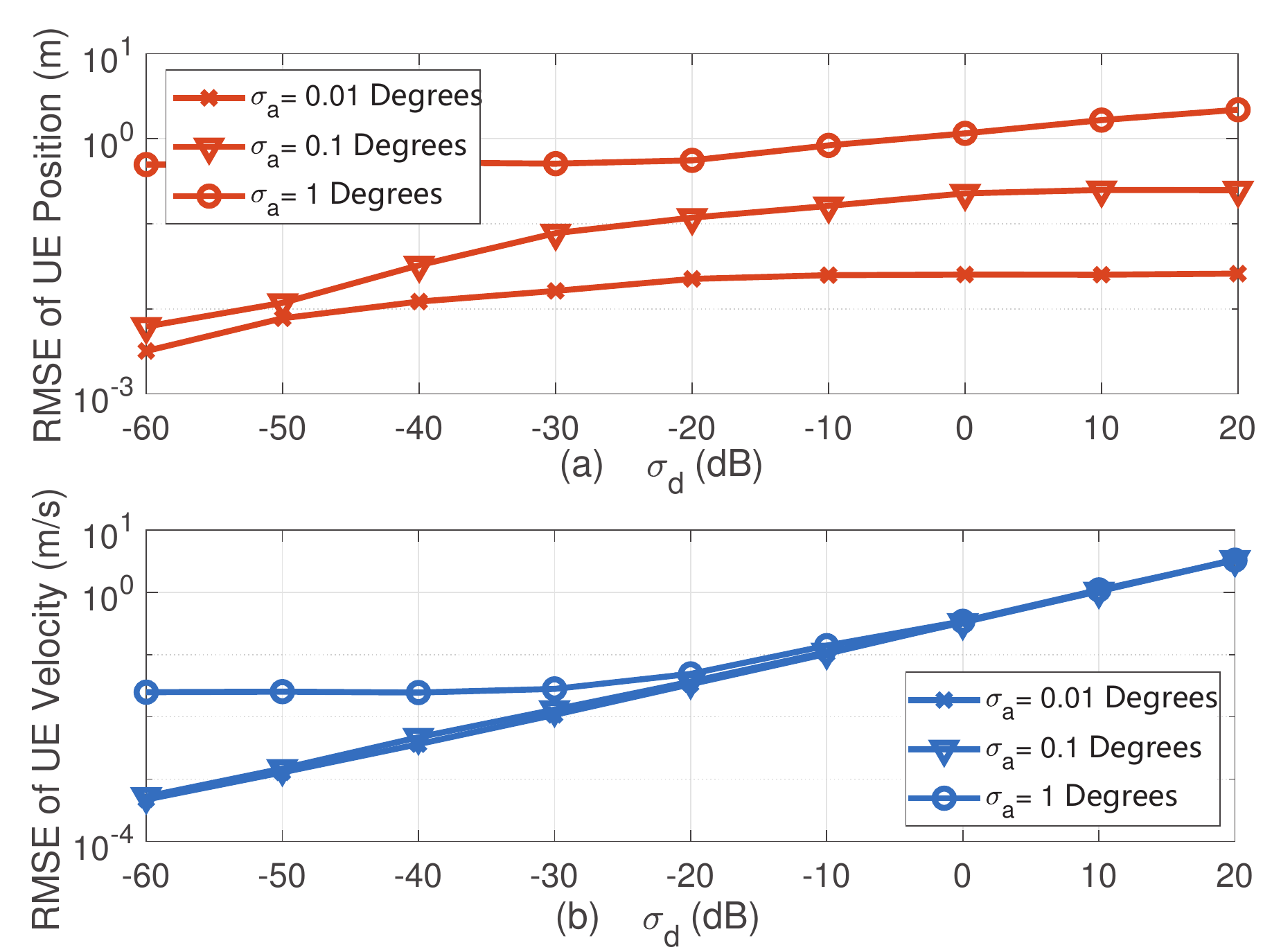}
		\captionsetup{font=footnotesize}
		\caption{RMSE as a function of $\sigma_d$ with a given $\sigma_a$ for (a) position of UE and (b) velocity of UE.}
		\label{fig:sigma}
	\end{minipage}
	\hspace{0.25cm}
	\begin{minipage}[t]{0.48\textwidth}
		\centering
		\includegraphics[scale=0.35,angle=0]{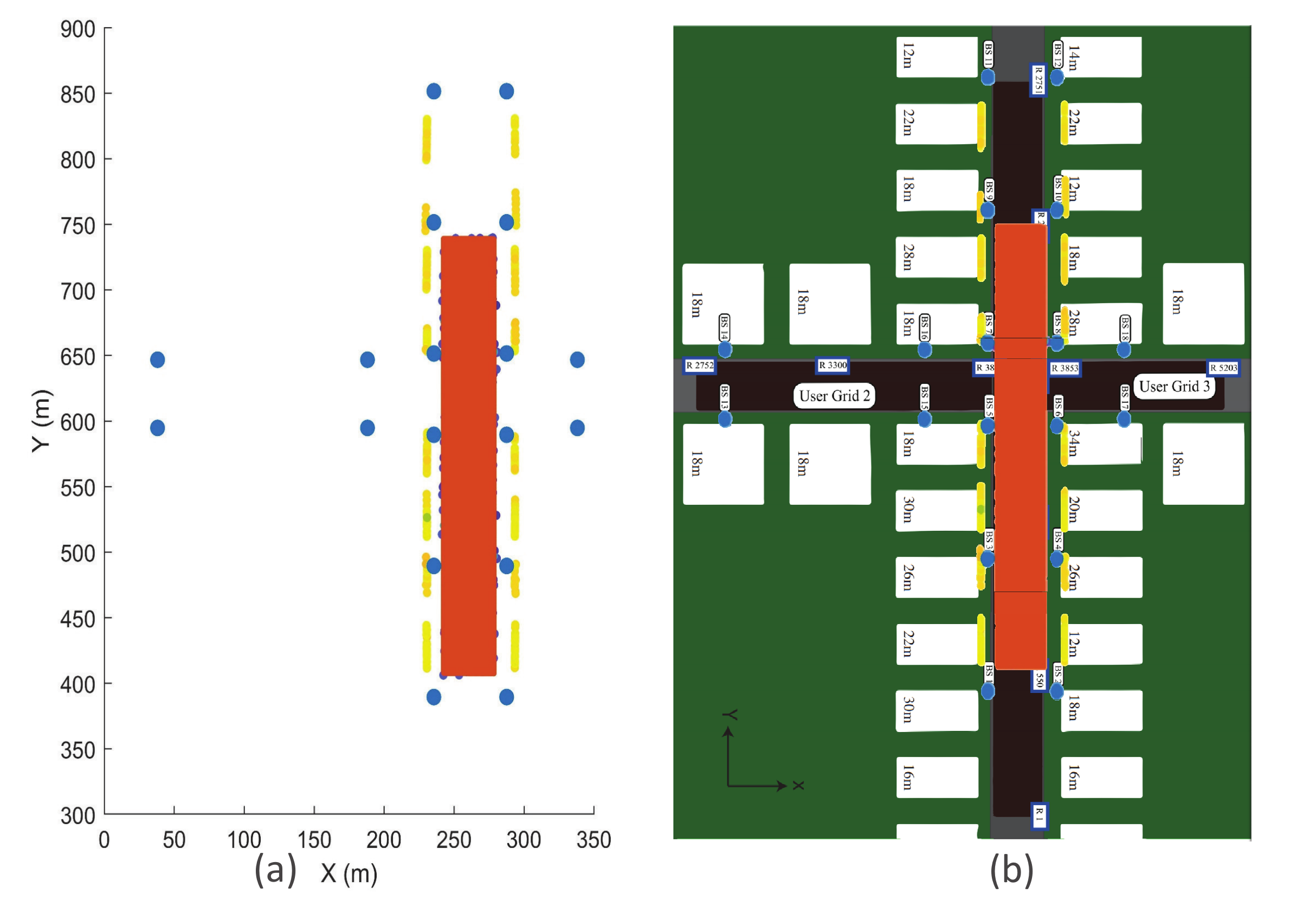}
		\captionsetup{font=footnotesize}
		\caption{Environment mapping result for the ray-tracing dataset. The blue points are the locations of 18 RRHs, the yellow points are the estimated positions of the scatterers, and the orange area is the UE moving area.}
		\label{mappingmatch}
	\end{minipage}
\end{figure}

\subsection{WLS-Based Environment Mapping}\label{ray}

In this section, we verify the effectiveness of the proposed environment mapping algorithm. We utilize a general dataset for mmWave massive MIMO constructed on the basis of the ray-tracing data from Remcom Wireless InSite \cite{deepmimo,raytracing} to verify the assumption of the single-bounce NLoS path, as this approach can accurately simulate real-world scenarios. The dataset represents an urban scenario with 18 RRHs and more than one million UEs. The deployment is shown in Fig. \ref{mappingmatch}(b), and the RRH locations are given in Table \ref{Tbs}.
\vspace{0.5cm}
\begin{table}[htbp]
	\footnotesize
	\centering
	\captionsetup{font=small}
	\caption{Locations of the RRHs in meters.}\label{Tbs}
	\begin{tabular}{|c|c|c|c|c|c|c|c|c|c|}
		\hline
		\hline
		& 1 & 2 & 3 & 4 & 5 & 6 & 7 & 8 & 9\\
		\hline
		x & 235.5042 & 287.5042 & 235.5042 & 287.5042 &  235.5042 & 287.5042 & 235.5042 & 287.5042 & 235.5042 \\
		y & 389.5038 & 389.5038 & 489.5038 & 489.5038 &  589.5038 & 589.5038 & 851.5038 & 851.5038 & 651.5038\\
		z & 6 & 6 & 6 & 6 & 6 & 6 & 6 & 6 & 6 \\      		
		\hline
		\hline
		& 10 & 11 & 12 & 13 & 14 & 15 & 16 & 17 & 18 \\
		\hline
		x & 287.5042 &  235.5042 & 287.5042 &  38.0751 & 38.0751 & 188.0751 & 188.0751 &  338.0751 & 338.0751 \\
		y & 651.5038 &  751.5038 & 751.5038 &  594.7361 & 646.7361 & 594.7361 & 646.7361 &  594.7361 & 646.7361 \\
		z & 6 & 6 & 6 & 6 & 6 & 6 & 6 & 6 & 6 \\
		\hline
	\end{tabular}
\end{table}

Let the UE move toward $1000$ different locations in the orange area of the street, as shown in Fig. \ref{mappingmatch}(a). The UE omni-directionally transmits a signal, and the RRHs receive the signal and obtain the hybrid measurements.
In our simulations,
we use the second and third strongest paths received by the first to twelfth RRHs for environment mapping.
The mapping result in Fig. \ref{mappingmatch} offers interesting findings.
First, the estimated locations of the scatterers match the position of the walls of the buildings, as illustrated in Fig. \ref{mappingmatch}(b).
Second, most of the second and third strongest paths in the ray-tracing dataset are verified to be the single-bounce NLoS paths, which also corroborates the assumption of the proposed environment mapping algorithm.

The physical interpretation of the proposed environment mapping algorithm has thus been proven. Subsequently, we will analyze algorithm performance by comparing the RMSE of the proposed environment mapping algorithm with that of the CRLB.
The simulation parameters are the same as those in Section \ref{sp}, including the number and position of the BSs and the location and velocity of the UE.
The unknown scatterer is located at $[50, 200, -70]^T$ in meters.
The numerical result is obtained from $1000$ independent Monte Carlo simulations.
As shown in Fig. \ref{emapping}, the proposed environment mapping algorithm can achieve the CRLB when $\rho \leq 5$ dB.
However,
the RMSE of the estimated location of the scatterer is larger than that of the UE, as shown in Fig. \ref{fig:crlb}(a),
because the number of measurements used in the environment mapping are less than that in the UE position estimation. The position estimation error of the UE will then proceed to the proposed environment mapping algorithm.

\begin{figure}[t]
	\centering
	\hspace{-0.4cm}
	\begin{minipage}[t]{0.48\textwidth}
		\centering
		\includegraphics[scale=0.4,angle=0]{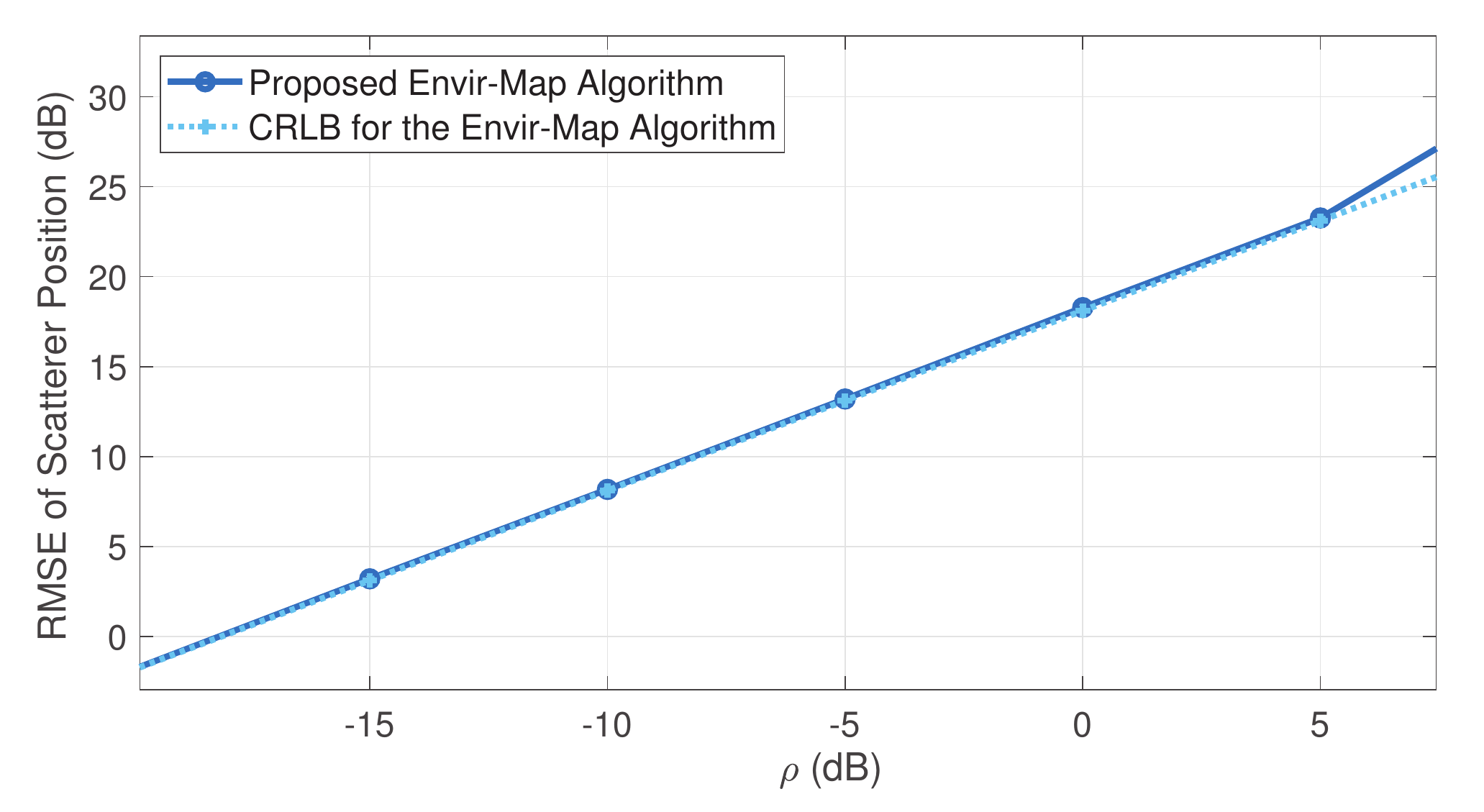}
		\captionsetup{font=footnotesize}
		\caption{Comparison of RMSEs of the proposed environment mapping algorithm and the CRLB.}
		\label{emapping}
	\end{minipage}
	\hspace{0.25cm}
	\begin{minipage}[t]{0.46\textwidth}
		\centering
		\includegraphics[scale=0.41,angle=0]{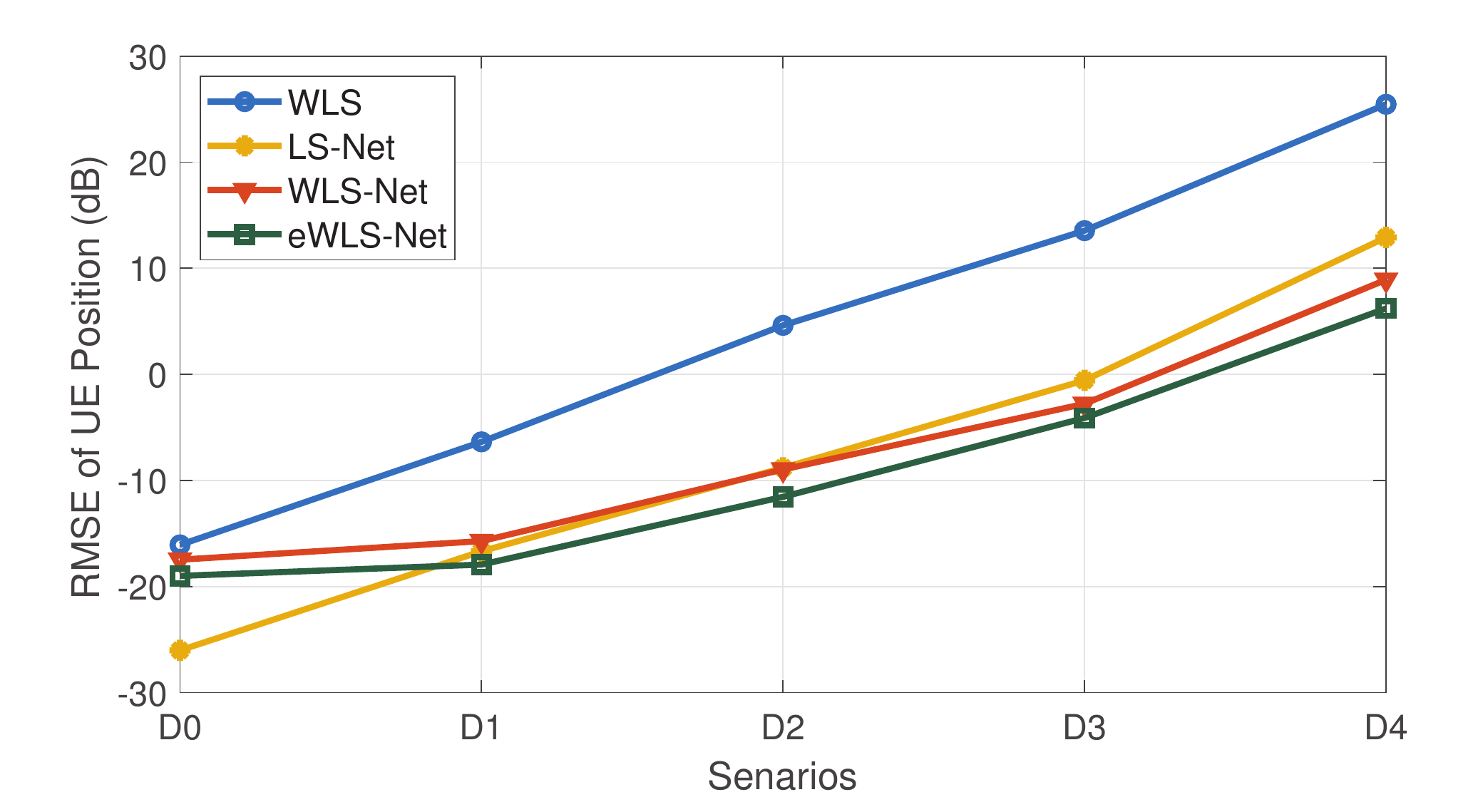}
		\captionsetup{font=footnotesize}
		\caption{RMSE curves of the UE position estimation for the WLS algorithm, WLS-Net, eWLS-Net, and LS-Net.}
		\label{1u}
	\end{minipage}
\end{figure}

\subsection{Neural Network-Assisted WLS Algorithm}
As mentioned in Sections \ref{sp} and \ref{ray}, when the noise level increases, the performance of the proposed WLS-based joint estimation and mapping algorithms will deviate from the CRLB  because the linear approximation of $\mathbf{e}$ becomes inaccurate.
We embed the neural networks into the proposed WLS algorithms to learn the residual vector $\mathbf{e}$, thus improving estimation accuracy. In this section, we will explore the performance of the proposed neural network-assisted WLS algorithm. 

The measurement errors in the ray-tracing dataset mentioned in Section \ref{ray} are relevant; hence, the neural network should be designed to learn the relationship and subsequently enhance localization performance. In the ray-tracing dataset, the standard deviations of the measurements related to TDoA, FDoA\footnote{There are no FDoA measurements given in the ray-tracing dataset. For each user, we generate its speed in a random way, and then calculate its corresponding FDoA measurements.}, azimuth AoA
($\phi$), and elevation AoA ($\theta$) are $0.0686$ (m), $0.0062$ (m/s), $1.675 \times10^{-6}$ (rad), and $1.239\times10^{-5}$ (rad), respectively.
The training, validation, and testing sets contain around $60000$, $20000$, and $20000$ samples, respectively. All testing samples are excluded from the training and validation samples.
The RMSE results of the FP algorithm\footnote{In the FP algorithm, the input and the neural network architecture is the same as the WLS-Net sub-Net 1 (except that the number of neurons in the output layer changes to 6), and the output of the neural network are directly the position and velocity of the UE.}, WLS (Section \ref{wls1}), WLS-Net (Section \ref{wlsnet}), and eWLS-Net (Section \ref{ewlsnet}) are given in Table \ref{T1}.
The result in Table \ref{T1} shows that the WLS-Net is more accurate in joint position and velocity estimation than the WLS algorithm.
The eWLS-Net can further improve the accuracy (in our simulations $L=10$).
The FP algorithm is the simplest to operate but has the worst accuracy.

\begin{table}[b]
	\vspace{0.4cm}
	\centering
	\hspace{-0.3cm}
	\begin{minipage}[t]{0.38\textwidth}
		\centering
		\footnotesize
		\captionsetup{font=footnotesize}
		\caption{RMSE of the FP, WLS, WLS-Net, and eWLS-Net with the ray-tracing dataset.}\label{T1}
		\begin{tabular}{c|c c}
			\hline
			\hline
			\multicolumn{2}{c|}{TDoA} & 0.0686 \\
			\multicolumn{2}{c|}{FDoA} & 0.0062 \\
			\multicolumn{2}{c|}{AoA ($\theta$)} & 1.675$\times10^{-6}$ \\
			\multicolumn{2}{c|}{AoA ($\phi$)} & 1.239$\times10^{-5}$ \\
			\hline
			\multirow{2}*{FP}& u & 0.1782 \\
			\cline{2-3}
			&\multicolumn{1}{>{\columncolor{mypink}}c}{v}  & \multicolumn{1}{>{\columncolor{mypink}}c}{0.2109} \\
			\cline{1-3}
			\multirow{2}*{WLS}& u & 0.0200 \\
			\cline{2-3}
			&\multicolumn{1}{>{\columncolor{mypink}}c}{v}  & \multicolumn{1}{>{\columncolor{mypink}}c}{0.0143} \\
			\cline{1-3}
			\multirow{2}*{WLS-Net}& u & 0.0180 \\
			\cline{2-3}
			& \multicolumn{1}{>{\columncolor{mypink}}c}{v} & \multicolumn{1}{>{\columncolor{mypink}}c}{0.0054} \\
			\cline{1-3}
			\multirow{2}*{eWLS-Net}& u & \textbf{0.0104} \\
			\cline{2-3}
			& \multicolumn{1}{>{\columncolor{mypink}}c}{v} & \multicolumn{1}{>{\columncolor{mypink}}c}{\textbf{0.0039}} \\
			\hline
		\end{tabular}
	\end{minipage}
	\hspace{0.15cm}
	\begin{minipage}[t]{0.6\textwidth}
		\centering
		\footnotesize
		\captionsetup{font=footnotesize}
		\caption{RMSE of the WLS algorithm, WLS-Net, eWLS-Net, and LS-Net in different scenarios.}\label{T2}
		\begin{tabular}{c|cccccc}
			\hline
			\hline
			\multicolumn{2}{c|}{scenario} & D0 & D1 & D2 & D3 & D4 \\
			\hline
			\multicolumn{2}{c|}{TDoA} & 0.1 & 1 & 10 & 100 & 1000 \\
			\multicolumn{2}{c|}{FDoA} & 0.01 & 0.1 & 1 & 10 & 100 \\
			\multicolumn{2}{c|}{AoA} & 0.001 & 0.01 & 0.1 & 1 & 10 \\
			\hline
			\multirow{2}*{WLS}& u & 0.0247 & 0.2311 & 2.8939 & 22.6304 & 351.0432 \\
			\cline{2-7}
			&\multicolumn{1}{>{\columncolor{mypink}}c}{v}  & \multicolumn{1}{>{\columncolor{mypink}}c}{0.0689} & \multicolumn{1}{>{\columncolor{mypink}}c}{0.8533} & \multicolumn{1}{>{\columncolor{mypink}}c}{11.2695} & \multicolumn{1}{>{\columncolor{mypink}}c}{21.1805}& \multicolumn{1}{>{\columncolor{mypink}}c}{56.3607} \\
			\cline{1-7}
			\multirow{2}*{WLS-Net}& u & 0.0179 & 0.0268 & 0.1268 & 0.5276 & 7.7910 \\
			\cline{2-7}
			& \multicolumn{1}{>{\columncolor{mypink}}c}{v} & \multicolumn{1}{>{\columncolor{mypink}}c}{0.0176} & \multicolumn{1}{>{\columncolor{mypink}}c}{0.0190} & \multicolumn{1}{>{\columncolor{mypink}}c}{0.4279} & \multicolumn{1}{>{\columncolor{mypink}}c}{1.5478} & \multicolumn{1}{>{\columncolor{mypink}}c}{14.1086} \\
			\cline{1-7}
			\multirow{2}*{eWLS-Net}& u & 0.0126 & \textbf{0.0161} & \textbf{0.0700} & \textbf{0.3880} & \textbf{4.1818} \\
			\cline{2-7}
			& \multicolumn{1}{>{\columncolor{mypink}}c}{v} & \multicolumn{1}{>{\columncolor{mypink}}c}{\textbf{0.0169}} & \multicolumn{1}{>{\columncolor{mypink}}c}{\textbf{0.0132}} & \multicolumn{1}{>{\columncolor{mypink}}c}{\textbf{0.4088}} & \multicolumn{1}{>{\columncolor{mypink}}c}{\textbf{1.2862}} & \multicolumn{1}{>{\columncolor{mypink}}c}{\textbf{12.7703}} \\
			\cline{1-7}
			\multirow{2}*{LS-Net}& u & \textbf{0.0025} & 0.0215 & 0.1314 & 0.5776 & 19.5489 \\
			\cline{2-7}
			& \multicolumn{1}{>{\columncolor{mypink}}c}{v} & \multicolumn{1}{>{\columncolor{mypink}}c}{0.0174} & \multicolumn{1}{>{\columncolor{mypink}}c}{0.0184} & \multicolumn{1}{>{\columncolor{mypink}}c}{0.4305} & \multicolumn{1}{>{\columncolor{mypink}}c}{1.5449} & \multicolumn{1}{>{\columncolor{mypink}}c}{14.2208} \\
			\hline
		\end{tabular}
	\end{minipage}
\end{table}

For an in-depth analysis of the performance of the neural network-assisted WLS algorithms, we will increase the measurement errors. By observing the ray-tracing dataset, we find that the measurement errors include a dominant part and a fluctuating part, in which the dominant part is the unknown fixed error and the fluctuating part is the Gaussian random error. We define five scenarios, from D0 to D4, by increasing the error variance of the two parts but also by keeping the ratio of the variance of the two parts unchanged. The standard deviations of the TDoA-, FDoA-, and AoA-related measurements for the fluctuating error are $0.0001$, $0.001$, and $0.001$ times of that for the dominant error, respectively. The standard deviations of the TDoA-, FDoA-, and AoA-related measurements for the dominant error in D0 are $0.1$ (m), $0.01$ (m/s), and $0.001$ (rad), respectively. For comparison, we define in this study the LS-Net-based joint position and velocity estimation algorithm. In particular, after obtaining the estimated residual vector ̂$\hat{\mathbf{e}}$ from the neural network (the same way as that implemented in the WLS-Net sub-Net 1), we deduct ̂$\hat{\mathbf{e}}$ from \eqref{20}. Then, we have
\vspace{-0.1cm}
\begin{equation}\label{200}
\tilde{\mathbf{h}} - \hat{\mathbf{e}} \doteq \tilde{\mathbf{G}}\mathbf{x}^{\circ}.
\vspace{-0.25cm}
\end{equation}
By directly applying the LS algorithm, we have
\vspace{-0.15cm}
\begin{equation}
\mathbf{x}=(\tilde{\mathbf{G}}^{T}\tilde{\mathbf{G}})^{-1}\tilde{\mathbf{G}}^{T}(\tilde{\mathbf{h}} - \hat{\mathbf{e}}).
\vspace{-0.25cm}
\end{equation}
The WLS algorithm, WLS-Net, eWLS-Net, and LS-Net are executed with the same datasets.
The RMSE results are shown in Table \ref{T2}.
For better illustration, we draw the RMSE curves of the UE position estimation for the WLS algorithm, WLS-Net, eWLS-Net, and LS-Net in Fig. \ref{1u}.
The RMSE curves of the UE velocity estimation have similar trends; hence, we omit the figure due to lack of space in this paper. In the figure, the performance of the algorithms based on the neural networks outperform the traditional WLS algorithm in the given simulation scenarios.
Interestingly,
when the measurement error is very small (in D0), the LS-Net performs best. However, as the measurement error increases (in D3 and D4), the advantage of the WLS-Net becomes apparent.
This phenomenon can be explained as follows: when the measurement error is small, $\mathbf{e}$ is small and the model in \eqref{200} is highly accurate. Then, the LS algorithm can be used to derive good results. Meanwhile, introducing the weight matrix $\mathbf{W}$ in the WLS-Net will introduce additional errors, which will deteriorate performance.
However, as the measurement error increases, the performance of the LS algorithm declines; by contrast, in the WLS algorithm, the weight matrix is
$\mathbf{W}=(\hat{\mathbf{e}}\hat{\mathbf{e}}^{T}+a\mathbf{I})^{-1}$, which contains both the information of the dominant error part ($\hat{\mathbf{e}}$ is the learned mean of the dominant error) and the information of the random error part ($a\mathbf{I}$ is the covariance matrix of the Gaussian random error).
Hence, the WLS-Net outperforms the LS-Net.

\begin{table}[b]
	\centering
	\footnotesize
	\captionsetup{font=footnotesize}
	\caption{RMSE of the FP algorithm, WLS algorithm, WLS-Net, and eWLS-Net in different scenarios.}\label{T0}
	\begin{tabular}{c|ccccc}
		\hline
		\hline
		\multicolumn{2}{c|}{scenario} & P1 & P2 & P3 & P4\\
		\hline
		\multirow{3}*{Ratio} & \multicolumn{1}{c|}{TDoA} & 0.0001 & 0.001 & 0.01 &  1\\
		& \multicolumn{1}{c|}{FDoA} & 0.001 & 0.01 & 0.1 &  1\\
		& \multicolumn{1}{c|}{AoA} & 0.001 & 0.01 & 0.1 &  1\\
		\hline
		\cline{1-6}
		\multirow{2}*{FP}& u & 0.2806 & 0.2481 & 0.2282 &  0.4874\\
		\cline{2-6}
		& \multicolumn{1}{>{\columncolor{mypink}}c}{v} & \multicolumn{1}{>{\columncolor{mypink}}c}{0.9568} & \multicolumn{1}{>{\columncolor{mypink}}c}{0.9812} & \multicolumn{1}{>{\columncolor{mypink}}c}{0.8919} &  \multicolumn{1}{>{\columncolor{mypink}}c}{0.8695}  \\
		\hline
		\multirow{2}*{WLS}& u & 0.2311 & 0.1620 & 0.2616 &  \textbf{0.2257} \\
		\cline{2-6}
		&\multicolumn{1}{>{\columncolor{mypink}}c}{v}  & \multicolumn{1}{>{\columncolor{mypink}}c}{0.8533} & \multicolumn{1}{>{\columncolor{mypink}}c}{0.7096} & \multicolumn{1}{>{\columncolor{mypink}}c}{0.5943} &    \multicolumn{1}{>{\columncolor{mypink}}c}{\textbf{0.7932}}\\
		\cline{1-6}
		\multirow{2}*{WLS-Net}& u & 0.0268 & 0.0233 & 0.0824 &  0.7989 \\
		\cline{2-6}
		& \multicolumn{1}{>{\columncolor{mypink}}c}{v} & \multicolumn{1}{>{\columncolor{mypink}}c}{0.0190} & \multicolumn{1}{>{\columncolor{mypink}}c}{0.0480} & \multicolumn{1}{>{\columncolor{mypink}}c}{0.1423} &  \multicolumn{1}{>{\columncolor{mypink}}c}{0.9886} \\
		\cline{1-6}
		\multirow{2}*{eWLS-Net}& u & \textbf{0.0161} & \textbf{0.0168} & \textbf{0.0813} &  0.7935\\
		\cline{2-6}
		& \multicolumn{1}{>{\columncolor{mypink}}c}{v} & \multicolumn{1}{>{\columncolor{mypink}}c}{\textbf{0.0132}} & \multicolumn{1}{>{\columncolor{mypink}}c}{\textbf{0.0455}} & \multicolumn{1}{>{\columncolor{mypink}}c}{\textbf{0.1402}} &  \multicolumn{1}{>{\columncolor{mypink}}c}{0.9851} \\
		\hline		
	\end{tabular}
\vspace{-0.5cm}
\end{table}
We analyze the performance of the proposed algorithms by increasing the ratio of the error standard deviation of the random part to that of the fixed part.
By fixing the standard deviation of the dominant error to the setting in D0, and by increasing the proportion of the random error, we can define four scenarios, from P1 to P4.
For example, for the standard deviation of the TDoA-related measurements in P3, the ratio of the random part to the fixed part is $0.01$.
In particular, we define P4 as a scenario in which the measurement error is completely the Gaussian random error.
The RMSE results are shown in Table \ref{T0}.
Given that the used FP algorithm is purely data driven and lacks the assistance of the geometric model, the performance of the FP algorithm in the simulations is worse than that of the proposed WLS algorithm.
Moreover, the proposed WLS algorithm will not learn the correlation between measurement errors; hence, the RMSE results of the WLS algorithm are similar in all cases, from P1 to P4. However, the WLS-Net can learn the dominant error. Therefore, when the proportion of the random part of the measurement error is small, the WLS-Net can surpass the WLS algorithm, and the eWLS-Net can further improve the accuracy. As the proportion of the random part increases, the ability of the neural networks decreases. Most especially in P4, for completely random measurement errors, the WLS algorithm performs the best.


\subsection{Time Resources}
In this section, we compare the time required by the different algorithms to jointly estimate the position and velocity of a single UE. The WLS algorithm needs $0.062251$ seconds.
The time needed by the WLS-Net consists of two parts.
The test time of the neural network is $1.595 \times 10^{-7}$ seconds,
and the time required to input the results of the neural network into the WLS algorithm and obtain the final estimation is $0.011337$ seconds.
Hence, the total time needed by the WLS-Net is $0.011337$ seconds, which is $18\%$ of the time needed by the WLS algorithm.
The eWLS-Net takes more time than the WLS-Net because the ensemble time is $0.0016$ seconds.
Thus, the total time needed by the eWLS-Net is $0.01294$ seconds,
which is $20.79\%$ of the time needed by the WLS algorithm.
This scheme is reasonable because the WLS algorithm requires several iterations, which is time consuming, whereas the WLS-Net and the eWLS-Net do not need to undergo such iterations.

\section{Conclusion}

This study considered the joint position and velocity estimation and environment mapping problem in the 3-D mmWave CRAN architecture. First, we embedded the cooperative localization into communications and established the joint position and velocity estimation model with hybrid AoA, TDoA, and FDoA measurements. Then, an efficient closed-form WLS solution was deduced and subsequently proven asymptotically unbiased. Second, we built the environment mapping model by exploiting the single-bounce NLoS paths and the estimated UE position. Additionally, we deduced the closed-form WLS solution for the environment mapping problem. The ray-tracing dataset was used to verify the physical demonstration of the single-bounce NLoS path assumption and the effectiveness of the proposed environment mapping algorithm. The simulation results indicate that the WLS-based joint estimation and mapping algorithm can achieve the CRLB under mild measurement noise and attain the desired decimeter-level accuracy.

Furthermore, a neural network-assisted WLS algorithm was proposed for SLAM by embedding the neural networks into the proposed WLS estimators to replace the linear approximation. This study is the first to combine the geometric model and the neural networks in 3-D SLAM methods in the literature. The combination harnesses both the powerful learning ability of the neural network and the robustness of the proposed geometric model. In addition, ensemble learning was introduced to improve performance. The public ray-tracing dataset was used in the simulations to test the performance of the neural network-assisted WLS algorithms, which was proven and can attain centimeter-level accuracy when the measurement error vector has some correlation pattern.

\begin{appendices}
\section{}\label{A}
In this section, we derive the noise-free matrix representation of the joint position and velocity estimation model.
Firstly, we derive $2(N_{a}-1)$ pseudo-linear TDoA and FDoA equations.
Rewriting \eqref{2} as $r_{n1}^\circ+r_1^\circ=r_{n}^\circ$ and squaring both sides yields
\vspace{-0.25cm}
\begin{equation} \label{8}
(r_{n1}^\circ)^2+2r_{n1}^\circ r_1^\circ=\mathbf{b}_n^T\mathbf{b}_n-\mathbf{b}_1^T\mathbf{b}_1-2(\mathbf{b}_n-\mathbf{b}_1)^T \mathbf{u}^\circ.
\vspace{-0.25cm}
\end{equation}
Equation \eqref{8} is pseudo-linear with respect to $\mathbf{u}^\circ$ and $r_1^\circ$,
where $r_1^\circ$ cannot be obtained directly from the channel measurement TDoA.
To eliminate $r_1^\circ$, we utilize the geometrical relationship
$\mathbf{u}^\circ-\mathbf{b}_1=r_1^\circ \mathbf{a}_1^\circ,$
The unit vector $\mathbf{a}_1^\circ$ possesses the properties:
$\mathbf{a}_1^{\circ T}\mathbf{a}_1^\circ=1$ and $\dot{\mathbf{a}}_1^{\circ T}\mathbf{a}_1^\circ=\mathbf{a}_1^{\circ T}\dot{\mathbf{a}}_1^\circ=0$.
Multiplying both sides of \eqref{8} by $\mathbf{a}_1^{\circ T}\mathbf{a}_1^\circ$ leads to
\vspace{-0.25cm}
\begin{equation} \label{11}
(r_{n1}^\circ)^2-2r_{n1}^\circ \mathbf{a}_1^{\circ T}\mathbf{b}_1-\mathbf{b}_n^T \mathbf{b}_n+\mathbf{b}_1^T \mathbf{b}_1=2[(\mathbf{b}_1-\mathbf{b}_n)^T-r_{n1}^\circ \mathbf{a}_1^{\circ T}]\mathbf{u}^\circ.
\vspace{-0.25cm}
\end{equation}
By taking the time derivative of \eqref{8}, we have
\vspace{-0.25cm}
\begin{equation} \label{12}
\dot{r}_{n1}^{\circ}{r}_{n1}^{\circ}+\dot{r}_{n1}^{\circ}r_1^\circ+r_{n1}^{\circ}\dot{r}_1^\circ=(\mathbf{b}_1-\mathbf{b}_n)^T \dot{\mathbf{u}}^\circ,
\vspace{-0.25cm}
\end{equation}
where $r_1^\circ$ and $\dot{r}_1^\circ$ cannot be obtained directly from the channel measurements TDoA and FDoA. We will eliminate them by the geometrical relationship
\vspace{-0.25cm}
\begin{equation} \label{13}
\dot{\mathbf{u}}^\circ=\dot{r}_1^\circ \mathbf{a}_1^\circ+r_1^\circ\dot{\mathbf{a}}_1^\circ.
\vspace{-0.25cm}
\end{equation}
Since $\mathbf{a}_1^{\circ T}\mathbf{a}_1^\circ=1$, according to \eqref{13}, we get
\vspace{-0.25cm}
\begin{equation} \label{14}
\dot{r}_{n1}^{\circ}{r}_{1}^{\circ}=\dot{r}_{n1}^{\circ}\mathbf{a}_1^{\circ T}(r_1^\circ \mathbf{a}_1^\circ)=\dot{r}_{n1}^{\circ}\mathbf{a}_1^{\circ T}(\mathbf{u}^\circ-\mathbf{b}_1),
\vspace{-0.25cm}
\end{equation}
and
\vspace{-0.25cm}
\begin{equation} \label{15}
{r}_{n1}^{\circ}\dot{r}_{1}^{\circ}={r}_{n1}^{\circ}\mathbf{a}_1^{\circ T}(\dot{r}_1^\circ \mathbf{a}_1^\circ)={r}_{n1}^{\circ}\mathbf{a}_1^{\circ T}(\dot{\mathbf{u}}^\circ-{r}_{1}^{\circ}\dot{\mathbf{a}}_1^\circ)={r}_{n1}^{\circ}\mathbf{a}_1^{\circ T}\dot{\mathbf{u}}^\circ.
\vspace{-0.25cm}
\end{equation}
Substituting \eqref{14} and \eqref{15} into \eqref{12}, we obtain
\vspace{-0.25cm}
\begin{equation} \label{16}
\dot{r}_{n1}^{\circ}{r}_{n1}^{\circ}-\dot{r}_{n1}^{\circ}\mathbf{a}_1^{\circ T}\mathbf{b}_1=-\dot{r}_{n1}^{\circ}\mathbf{a}_1^{\circ T}\mathbf{u}^\circ+[(\mathbf{b}_1-\mathbf{b}_n)^T-{r}_{n1}^{\circ}\mathbf{a}_1^{\circ T}]\dot{\mathbf{u}}^\circ.
\vspace{-0.25cm}
\end{equation}
According to (\ref{11}) and (\ref{16}), for $n=1,\ldots, N_{a}$, $2(N_{a}-1)$ pseudo-linear TDoA and FDoA equations are obtained.
Then, we derive $2N_{a}$ AoA equations for $n=1,2,\ldots, N_{a}$, given by
\vspace{-0.25cm}
\begin{equation} \label{17}
\mathbf{c}_n^{\circ T}\mathbf{b}_n=\mathbf{c}_n^{\circ T}\mathbf{u}^\circ, \ \ \mathbf{d}_n^{\circ T}\mathbf{b}_n=\mathbf{d}_n^{\circ T}\mathbf{u}^\circ.
\vspace{-0.25cm}
\end{equation}
Stacking equations \eqref{11}, \eqref{16} and  \eqref{17} yields
$\mathbf{h}=\mathbf{G}\mathbf{x}^\circ$.

\section{}\label{B}
In this section, we provide a particular choice of $\mathbf{W}$ that minimizes the variance of $\mathbf{x}$.
In view of the nonlinearity of $\mathbf{e}$, it is difficult to get the weighting matrix $\mathbf{W}$ in general. According to \cite{KAY}, the weighting matrix that minimizes the Frobenius norm of the covariance matrix of $\mathbf{x}$ is $\mathbf{W}=(\mathbb{E}\{\mathbf{e}\mathbf{e}^{T}\})^{-1}$.
Firstly, we approximate $\mathbf{e}$ up to the linear noise term.
According to \eqref{191} and \eqref{20}, we get
\vspace{-0.25cm}
\begin{equation}\label{e0}
\mathbf{e}=(\tilde{\mathbf{h}}-\tilde{\mathbf{G}}\mathbf{x}^{\circ})-(\mathbf{h}-\mathbf{G}\mathbf{x}^\circ).
\vspace{-0.25cm}
\end{equation}
Note that $\mathbf{h}$ and $\mathbf{G}$ are expressed by $\{r_{i1}^{\circ},\dot{r}_{i1}^{\circ},\phi_j^{\circ},\theta_j^{\circ}\}$, and $\tilde{\mathbf{h}}$ and $\tilde{\mathbf{G}}$ are expressed by $\{r_{i1}=r_{i1}^{\circ}+\Delta r_{i1}$, $\dot{r}_{i1}=\dot{r}_{i1}^{\circ}+\Delta \dot{r}_{i1}$, $\phi_j=\phi_j^{\circ}+\Delta \phi_j$,
$\theta_j=\theta_j^{\circ}+\Delta\theta_j\}$, for $i=2,\ldots,N_{a}$ and $j=1,\ldots,N_{a}$.
For the differentiable function $f(x_1,x_1,\ldots,x_n)$ on the variables $x_1,x_1,\ldots,x_n$,
there holds
\vspace{-0.25cm}
\begin{equation}\label{23}
f(x_1+\Delta x_1,x_2+\Delta x_2,\ldots,x_n+\Delta x_n)\!-\!f(x_1,x_2,\ldots,x_n)
\frac{\partial f}{\partial x_1}\!\Delta x_1\!+\!\frac{\partial f}{\partial x_2}\!\Delta x_2
\!+\!\ldots\!+\!\frac{\partial f}{\partial x_n}\!\Delta x_n\!+\!o(\rho) \  as \ \rho\!\rightarrow\! 0,
\vspace{-0.15cm}
\end{equation}
where $\rho=\sqrt{(\Delta x_1)^2+(\Delta x_2)^2+\ldots+(\Delta x_n)^2}$.
Applying (\ref{23}) with (\ref{e0}), firstly, for $i=2,\ldots,N_{a}$, we yield
{\begingroup\makeatletter\def\f@size{11}\check@mathfonts
	\def\maketag@@@#1{\hbox{\m@th\normalsize\normalfont#1}}\setlength{\arraycolsep}{0.0em}\setlength{\arraycolsep}{0.0em}
	\begin{eqnarray*}\label{24}
\begin{aligned}
\hspace{-0.13cm}
\mathbf{e}(2i\!-\!3)
&\dot{=}[2r_{i1}^{\circ}\!\!+\!\!2\mathbf{a}_1^{\circ T}\!(\mathbf{u}^{\circ}\!\!-\!\mathbf{b}_1)]\Delta r_{i1}
\!\!+\!\!\left[-2r_{i1}^{\circ}\frac{\partial \mathbf{a}_1^{\circ T}}{\partial{\phi_1^{\circ}}}\mathbf{b}_1
\!\!+\!\!2r_{i1}^{\circ}\frac{\partial \mathbf{a}_1^{\circ T}}{\partial{\phi_1^{\circ}}}\mathbf{u}^{\circ}\right]\!\!
\Delta\phi_1
\!\!+\!\!\left[-2r_{i1}^{\circ}\frac{\partial \mathbf{a}_1^{\circ T}}{\partial{\theta_1^{\circ}}}\mathbf{b}_1
+2r_{i1}^{\circ}\frac{\partial \mathbf{a}_1^{\circ T}}{\partial{\theta_1^{\circ}}}\mathbf{u}^{\circ}\right]\!\!
\Delta\theta_1,
\end{aligned}
\end{eqnarray*}\setlength{\arraycolsep}{5pt}\endgroup}where $2r_{i1}^{\circ}\frac{\partial \mathbf{a}_1^{\circ T}}{\partial{\phi_1^{\circ}}}(\mathbf{u}^{\circ}-\mathbf{b}_1)=2r_{i1}^{\circ}r_{1}^{\circ} \frac{\partial \mathbf{a}_1^{\circ T}}{\partial{\phi_1^{\circ}}}\mathbf{a}_1^{\circ} = 0$, hence, we have
\vspace{-0.15cm}
\begin{equation}\label{241}
\mathbf{e}(2i-3)=2r_i^{\circ} \Delta r_{i1}.
\vspace{-0.3cm}
\end{equation}
Secondly, we have
{\begingroup\makeatletter\def\f@size{11}\check@mathfonts
	\def\maketag@@@#1{\hbox{\m@th\normalsize\normalfont#1}}\setlength{\arraycolsep}{0.0em}\setlength{\arraycolsep}{0.0em}
	\begin{eqnarray*}
\begin{aligned}
\mathbf{e}(2i\!-\!2)&\dot{=}(\dot{r}_{i1}^{\circ}+\mathbf{a}_1^{\circ T}\dot{\mathbf{u}}^{\circ})\Delta r_{i1}+
\left[r_{i1}^{\circ}+\mathbf{a}_1^{\circ T}(\mathbf{u}^{\circ}-\mathbf{b}_1)\right]\!\!\Delta\dot{r}_{i1}
+\!\!\left[-\dot{r}_{i1}^{\circ}\frac{\partial \mathbf{a}_1^{\circ T}}{\partial{\phi_1^{\circ}}}\mathbf{b}_1
\!+\!\dot{r}_{i1}^{\circ}\frac{\partial \mathbf{a}_1^{\circ T}}{\partial{\phi_1^{\circ}}}\mathbf{u}^{\circ}
\!+\!r_{i1}^{\circ}\frac{\partial \mathbf{a}_1^{\circ T}}{\partial{\phi_1^{\circ}}}\dot{\mathbf{u}}^{\circ}\right]
\!\!\Delta\phi_1\\
&+\!\!\left[-\dot{r}_{i1}^{\circ}\frac{\partial \mathbf{a}_1^{\circ T}}{\partial{\theta_1^{\circ}}}\mathbf{b}_1
\!+\!\dot{r}_{i1}^{\circ}\frac{\partial \mathbf{a}_1^{\circ T}}{\partial{\theta_1^{\circ}}}\mathbf{u}^{\circ}
\!+\!r_{i1}^{\circ}\frac{\partial \mathbf{a}_1^{\circ T}}{\partial{\theta_1^{\circ}}}\dot{\mathbf{u}}^{\circ}\right]
\!\!\Delta\theta_1.
\end{aligned}
\end{eqnarray*}\setlength{\arraycolsep}{5pt}\endgroup}By some tedious calculations, we obtain
\vspace{-0.25cm}
\begin{equation}\label{25i}
\mathbf{e}(2i-2)\dot{=} \dot{r}_{i}^{\circ}\Delta r_{i1}+r_{i}^{\circ}\Delta \dot{r}_{i1}+r_{1}^{\circ}r_{i1}^{\circ}\cos^{2}\theta_1^{\circ}\dot{\phi}_1^\circ\Delta\phi_1
+r_{1}^{\circ}r_{i1}^{\circ}\dot{\theta}_1^\circ\Delta\theta_1.
\vspace{-0.25cm}
\end{equation}
Thirdly, for $j=1,\ldots,N_{a}$, we have
\begin{equation}\label{26}
\mathbf{e}(2N_{a}-3+2j)\dot{=}\left(\frac{\partial{\mathbf{c}_j^{\circ T}}}{\partial{\phi_j^{\circ}}}\mathbf{b}_j
-\frac{\partial{\mathbf{c}_j^{\circ T}}}
{\partial{\phi_j^{\circ}}}\mathbf{u}^{\circ}\right)\Delta\phi_j=r_j^{\circ}\cos\theta_j^{\circ}\Delta\phi_j.
\end{equation}
Then, we get
\vspace{-0.25cm}
\begin{equation}\label{28}
\begin{aligned}
\mathbf{e}(2N_{a}-2+2j)&\dot{=}\frac{\partial{\mathbf{d}_j^{\circ T}}}{\partial{\phi_j^{\circ}}}
\left(\mathbf{b}_j
-\mathbf{u}^{\circ}\right)\Delta\phi_j +\frac{\partial{\mathbf{d}_j^{\circ T}}}{\partial{\theta_j^{\circ}}}
\left(\mathbf{b}_j
-\mathbf{u}^{\circ}\right)\Delta\theta_j=r_j^{\circ}\Delta\theta_j.
\end{aligned}
\vspace{-0.25cm}
\end{equation}
Finally, transform the expressions \eqref{241}, \eqref{25i}, \eqref{26}, and \eqref{28} for $i=2,\ldots,N_{a}$ and $j=1,\ldots,N_{a}$ into matrix representation, we obtain
$\mathbf{e} \doteq \mathbf{B}\Delta \mathbf{m}$,
where
{\begingroup\makeatletter\def\f@size{11}\check@mathfonts
	\def\maketag@@@#1{\hbox{\m@th\normalsize\normalfont#1}}\setlength{\arraycolsep}{0.0em}\setlength{\arraycolsep}{0.0em}
	\begin{eqnarray}\label{b}
	\mathbf{B}=\begin{bmatrix}
	2r_2^{\circ}&0&\ldots&0&0&0&0&\ldots&0&0\\
	\dot{r}_2^{\circ}&r_2^{\circ}&\ldots&0&0&a_2
	&b_2&\ldots&0&0\\
	\vdots&\vdots&\vdots&\vdots&\vdots&\vdots&\vdots&\vdots&\vdots&\vdots\\
	0&0&\ldots&2r_N^{\circ}&0&0&0&\ldots&0&0\\
	0&0&\ldots&\dot{r}_N^{\circ}&r_N^{\circ}&a_N
	&b_N&\ldots&0&0\\
	0&0&\ldots&0&0&r_1^{\circ}\cos\theta_1^{\circ}&0&\ldots&0&0\\
	0&0&\ldots&0&0&0&r_1^{\circ}&\ldots&0&0\\
	\vdots&\vdots&\vdots&\vdots&\vdots&\vdots&\vdots&\vdots&\vdots&\vdots\\
	0&0&\ldots&0&0&0&0&\ldots&r_N^{\circ}\cos\theta_N^{\circ}&0\\
	0&0&\ldots&0&0&0&0&\ldots&0&r_N^{\circ}\\
	\end{bmatrix},
	\end{eqnarray}\setlength{\arraycolsep}{5pt}\endgroup}in which $a_2 = r_1^{\circ}r_{21}^{\circ}\dot{\phi}_1^{\circ}\cos^2\theta_1^{\circ}$, $a_N = r_1^{\circ}r_{N1}^{\circ}\dot{\phi}_1^{\circ}\cos^2\theta_1^{\circ}$, $b_2 = r_1^{\circ}r_{21}^{\circ}\dot{\theta}_1^{\circ}$, $b_N = r_1^{\circ}r_{N1}^{\circ}\dot{\theta}_1^{\circ}$, and
$\dot{\phi}_1^{\circ}$ and $\dot{\theta}_1^{\circ}$ are the time derivative of
\eqref{5} with $n=1$, we have
\begin{equation} \label{6}
\dot{\phi}_1^\circ=\frac{\mathbf{a}_1^{\circ T}\dot{\mathbf{u}}^\circ}{r_1^\circ\cos\theta_1^\circ}, \
\dot{\theta}_1^\circ=\frac{\dot{\mathbf{u}}^{\circ T}\mathbf{b}_1^\circ}{r_1^\circ}.
\end{equation}
Since the noise vector $\Delta \mathbf{m}$ is a zero-mean Gaussian vector with
covariance matrix $\mathbf{Q}$, that is,
$\mathbb{E}\{\Delta \mathbf{m}\}=\mathbf{0}$ and $\mathbb{E}\{\Delta \mathbf{m} \Delta \mathbf{m}^T\}=\mathbf{Q}$.
Hence, the distribution of $\Delta \mathbf{m}$ implies the asymptotic distribution of $\mathbf{e}$. The expectation of $\mathbf{e}$ is
$\mathbb{E}\{\mathbf{e}\}\!\doteq\! \mathbb{E}\{\mathbf{B}\Delta \mathbf{m}\}\!=\!\mathbf{B}\mathbb{E}\{\Delta \mathbf{m}\}\!=\!\mathbf{0}$,
and the covariance matrix of $\mathbf{e}$ is $
\mbox{cov}(\mathbf{e})=\mathbb{E}\left\{(\mathbf{e}-\mathbb{E}(\mathbf{e}))(\mathbf{e}-\mathbb{E}(\mathbf{e}))^T\right\}\doteq \mathbb{E}\{\mathbf{e} \mathbf{e}^T\}=\mathbf{B}\mathbf{Q}\mathbf{B}^T$.
Therefore, from $\mathbf{W}=(\mathbb{E}\{\mathbf{e}\mathbf{e}^{T}\})^{-1}$, the weighting matrix can be easily calculated as
$\mathbf{W}=(\mathbf{B}\mathbf{Q}\mathbf{B}^T)^{-1}$.

\section{}\label{D}
In this section, we will calculate the partial derivatives appearing in (\ref{35}), respectively.
Firstly, from (\ref{2}) and (\ref{1}), for $i=2,\ldots,N_{a}$, we have
$r_{i1}^{\circ}\!=\!\sqrt{(\mathbf{u}^{\circ}\!-\!\mathbf{b}_i)^{T}(\mathbf{u}^{\circ}\!-\!\mathbf{b}_i)}\!-\!\sqrt{(\mathbf{u}^{\circ}\!-\!\mathbf{b}_1)^T(\mathbf{u}^{\circ}\!-\!\mathbf{b}_1)}$,
hence, we obtain
\begin{equation}\label{p1}
\dfrac{\partial r_{i1}^{\circ}}{\partial \mathbf{u}^{\circ T}}=\dfrac{(\mathbf{u}^{\circ}-\mathbf{b}_i)^T}{r_i^{\circ}}-\dfrac{(\mathbf{u}^{\circ}-\mathbf{b}_1)^T}{r_1^\circ}, \quad \dfrac{\partial r_{i1}^{\circ}}{\partial \dot{\mathbf{u}}^{\circ T}}=\mathbf{0}.
\end{equation}
Secondly, from (\ref{3}) and (\ref{4}), we have
$\dot{r}_{i1}^{\circ}=\dot{\mathbf{u}}^{\circ T}(\mathbf{u}^{\circ}-\mathbf{b}_i)/r^{\circ}_i-\dot{\mathbf{u}}^{\circ T}(\mathbf{u}^{\circ}-\mathbf{b}_1)/r^{\circ}_1
$,
then, we get
\begin{equation}\label{p2}
\begin{aligned}
\dfrac{\partial\dot{r}_{i1}^{\circ}}{\partial \mathbf{u}^{\circ T}}&=\dfrac{\dot{r}_1^{\circ}(\mathbf{u}^{\circ}-\mathbf{b}_1)^T}{(r_1^{\circ })^2}-\dfrac{\dot{r}_i^{\circ}(\mathbf{u}^{\circ}-\mathbf{b}_i)^{T}}{(r_i^{\circ })^2}+\dfrac{\dot{\mathbf{u}}^{\circ T}}{r_i^{\circ}}-\dfrac{\dot{\mathbf{u}}^{\circ T}}{r_1^{\circ}}, \\
\dfrac{\partial\dot{r}_{i1}^{\circ}}{\partial \dot{\mathbf{u}}^{\circ T}}&=\dfrac{(\mathbf{u}^{\circ}-\mathbf{b}_i)^T}{r_i^{\circ}}-\dfrac{(\mathbf{u}^{\circ}-\mathbf{b}_1)^T}{r^{\circ}_1}.
\end{aligned}
\end{equation}
Thirdly, according to (\ref{17}), for $j=1,\ldots,N_{a}$, we get
\begin{equation}\label{p}
(\mathbf{b}_j-\mathbf{u}^{\circ})^T\dfrac{\partial \mathbf{c}^{\circ}_j}{\partial \mathbf{u}^{\circ T}}=\mathbf{c}^{\circ T}_j,
\end{equation}
since $\mathbf{a}_j^{\circ T}[\cos\phi^{\circ}_j,\sin\phi^{\circ}_j,0]^T=\cos\theta^{\circ}_j$, we yield
\begin{eqnarray}\label{q}
	(\mathbf{b}_j -\mathbf{u}^{\circ})^T\dfrac{\partial \mathbf{c}^{\circ}_j}{\partial \mathbf{u}^{\circ T}} =-r^{\circ}_j\mathbf{a}_j^{\circ T}\!\dfrac{\partial \mathbf{c}^{\circ}_j}{\partial \phi^{\circ}_j} \dfrac{\partial \phi^{\circ}_j}{\partial \mathbf{u}^{\circ T}}= r^{\circ}_j\!\cos\theta_j^{\circ}\dfrac{\partial \phi^{\circ}_j}{\partial \mathbf{u}^{\circ T}},
\end{eqnarray}
substituting \eqref{q} into \eqref{p}, we have
\begin{equation}\label{p3}
\dfrac{\partial \phi^{\circ}_j}{\partial \mathbf{u}^{\circ T}}=\dfrac{\mathbf{c}^{\circ T}_j}{r^{\circ}_j\cos\theta_j^{\circ}}.
\end{equation}
Similarly, from \eqref{17}, we obtain
$
(\mathbf{u}^{\circ} \!-\!\mathbf{b}_j)^T\dfrac{\partial \mathbf{d}^{\circ}_j}{\partial \mathbf{u}^{\circ T}}+\mathbf{d}^{\circ T}_j=\mathbf{0},
$
that is,
$
(\mathbf{u}^{\circ} \!-\!\mathbf{b}_j)^T[\dfrac{\partial  \mathbf{d}^{\circ}_j}{\partial \theta^{\circ}_j}\dfrac{\partial\theta^{\circ}_j}{\partial \mathbf{u}^{\circ T}}+\dfrac{\partial \mathbf{d}^{\circ}_j}{\partial \phi^{\circ}_j}\dfrac{\partial\phi^{\circ}_j}{\partial \mathbf{u}^{\circ T}}]=-\mathbf{d}^{\circ T}_j,
$
which is equivalent to
$
r^{\circ}_j \mathbf{a}^{\circ T}_j\dfrac{\partial \mathbf{d}^{\circ}_j}{\partial \theta^{\circ}_j}\dfrac{\partial\theta^{\circ}_j}{\partial \mathbf{u}^{\circ T}}+r^{\circ}_j\mathbf{a}^{\circ T}_j\dfrac{\partial \mathbf{d}^{\circ}_j}{\partial \phi^{\circ}_j}\dfrac{\partial\phi^{\circ}_j}{\partial \mathbf{u}^{\circ T}}=-\mathbf{d}^{\circ T}_j,
$
since
$\mathbf{a}^{\circ T}_j\dfrac{\partial \mathbf{d}^{\circ}_j}{\partial \theta^{\circ}_j}=-1$ and
$ \mathbf{a}^{\circ T}_j\dfrac{\partial \mathbf{d}^{\circ}_j}{\partial \phi^{\circ}_j}=0$,
we get
\begin{equation}\label{p5}
\dfrac{\partial \theta^{\circ}_j}{\partial \mathbf{u}^{\circ T}}=\dfrac{\mathbf{d}^{\circ T}_j}{r^{\circ}_j}.
\end{equation}
Finally, it is obvious that
\begin{equation}\label{p6}
\dfrac{\partial \phi^{\circ}_j}{\partial \dot{\mathbf{u}}^{\circ T}}=\mathbf{0},
\ \ \dfrac{\partial \theta^{\circ}_j}{\partial \dot{\mathbf{u}}^{\circ T}}=\mathbf{0}.
\end{equation}
Therefore, the partial derivatives are given in \eqref{p1}, \eqref{p2}, \eqref{p3}, \eqref{p5}, and \eqref{p6}.

\section{}\label{E}
In this section, we prove that $\mbox{ cov}(\mathbf{x})\!\doteq\!\!\mbox{ CRLB}(\mathbf{x}^{\circ})$,
which holds, if and only if
\begin{equation}\label{101}
\mathbf{B}\mathbf{B}_1=\mathbf{G},
\end{equation}
where $\mathbf{G}$, $\mathbf{B}_{1}$, and $\mathbf{B}$ are defined by (\ref{191}), (\ref{99}) and (\ref{b}), respectively.
By direct matrix multiplications (\ref{101}) can be verified.
However, its proof relies on the following two key identities, for $i=2,3,\ldots,N_{a}$,
\begin{equation}\label{aa}
(a):\ r_{i}^{\circ}\!\left[\!\frac{(\mathbf{u}^{\circ}\!-\!\mathbf{b}_{i})^T}{r_{i}^{\circ}}\!-\!\frac{(\mathbf{u}^{\circ}\!-\!\mathbf{b}_{1})^{T}}{r_{1}^{\circ}}\!\right]\!\!=\!\!(\mathbf{b}_{1}\!-\!\mathbf{b}_{i})^{T}\!\!\!-\!r_{i1}^{\circ}\mathbf{a}_{1}^{\circ T}\!,
\vspace{-0.25cm}
\end{equation}
{\begingroup\makeatletter\def\f@size{11}\check@mathfonts
	\def\maketag@@@#1{\hbox{\m@th\normalsize\normalfont#1}}\setlength{\arraycolsep}{0.0em}\setlength{\arraycolsep}{0.0em}
	\begin{eqnarray}\label{bb}
\begin{split}
(b):&\dot{r}_{i}^{\circ}\!\!\left[\!\frac{(\mathbf{u}^{\circ}\!\!-\!\mathbf{b}_{i})^{T}}{r_{i}^{\circ}}\!-\!\frac{(\mathbf{u}^{\circ}\!\!-\!\mathbf{b}_{1})^{T}}{r_{1}^{\circ}}\!\right]\!\!+\!r_{i}^{\circ}\!\bigg[\!\frac{\dot{r}_{1}^{\circ}\!(\mathbf{u}^{\circ}\!\!-\!\mathbf{b}_{1})\!^{T}}{(r_{1}^{\circ})^{2}}
-\frac{\dot{r}_{i}^{\circ}\!(\mathbf{u}^{\circ}\!\!-\!\mathbf{b}_{i})\!^{T}}{(r_{i}^{\circ})^{2}}
\!+\!\frac{\dot{\mathbf{u}}^{\circ T}}{r_{i}^{\circ}}-\frac{\dot{\mathbf{u}}^{\circ T}}{r_{1}^{\circ}}\bigg]
\!+\!r_{i1}^{\circ}\dot{\phi}_{1}^{\circ}\!\cos\theta_{1}^{\circ}\mathbf{c}_{1}^{\circ T}\!+\!r_{i1}^{\circ}\dot{\theta}_{1}^{\circ}\mathbf{d}_{1}\!^{\circ T}\!\!
\\
&=\!-\dot{r}_{i1}^{\circ}\mathbf{a}_{1}^{\circ T}\!.
\end{split}
\end{eqnarray}\setlength{\arraycolsep}{5pt}\endgroup}Firstly, we prove (a), which is equivalent to
\begin{align}\label{140}
-r_{i1}^{\circ}\mathbf{a}_{1}^{\circ T}\! =\! (\mathbf{u}^{\circ}\!-\!\mathbf{b}_{i})^T\!-\!(\mathbf{b}_{1}\!-\!\mathbf{b}_{i})^{T}\!-\!r_{i}^{\circ}/r_{1}^{\circ}(\mathbf{u}^{\circ}\!-\!\mathbf{b}_{1})^{T}.
\end{align}
Since the right-hand side of (\ref{140}) equals to
\begin{equation}
(\mathbf{u}^{\circ}\!\!-\!\mathbf{b}_{1})^T\!\!\!-\!r_{i}^{\circ}/r_{1}^{\circ}(\mathbf{u}^{\circ}\!\!-\!\mathbf{b}_{1})^{T}
\!\!\!=\!r_{1}^{\circ}\mathbf{a}_{1}^{\circ T} \!\!-\!r_{i}^{\circ}\mathbf{a}_{1}^{\circ T}
\!=\!-r_{i1}^{\circ}\mathbf{a}_{1}^{\circ T},
\end{equation}
therefore, (a) holds. Then, we prove (b), and the left-hand side of \eqref{bb} equals to
\begin{equation}
\begin{split}
&\dot{r}_{i}^{\circ}\!(\!\mathbf{a}_{i}^{\circ T}\!\!-\!\mathbf{a}_{1}^{\circ T}\!)-\dot{r}_{i}^{\circ}\mathbf{a}_{i}^{\circ T}\!\!+\frac{r_{i}^{\circ}\dot{r}_{1}^{\circ}}{r_{1}^{\circ}}\mathbf{a}_{1}^{\circ T}\!\! +\dot{\mathbf{u}}^{\circ T}\!\!-\frac{r_{i}^{\circ}\dot{\mathbf{u}}^{\circ T}}{r_{1}^{\circ}}+r_{i1}^{\circ}\dot{\phi}_{1}^{\circ}\frac{\partial \mathbf{a}_{1}^{\circ T} }{\partial\phi_{1}^{\circ}}+r_{i1}^{\circ}\dot{\theta}_{1}^{\circ}\frac{\partial \mathbf{a}_{1}^{\circ T} }{\partial\theta_{1}^{\circ}}\\
=&\!(\frac{r_{i}^{\circ}\dot{r}_{1}^{\circ}}{r_{1}^{\circ}}-\dot{r}_{i}^{\circ})\mathbf{a}_{1}^{\circ T}\!\!\! -\!\frac{r_{i1}^{\circ}}{r_{1}^{\circ}}\dot{\mathbf{u}}^{\circ T}\!\!+\!r_{i1}^{\circ}(\dot{\phi}_{1}^{\circ}\frac{\partial \mathbf{a}_{1}^{\circ T} }{\partial\phi_{1}^{\circ}}+\dot{\theta}_{1}^{\circ}\frac{\partial \mathbf{a}_{1}^{\circ T} }{\partial\theta_{1}^{\circ}}\!),
\end{split}
\end{equation}
since $\dot{\mathbf{u}}^\circ=\dot{r}_1^\circ \mathbf{a}_1^\circ+r_1^\circ\dot{\mathbf{a}}_1^\circ$ and $\dot{\phi}_{1}^{\circ}\frac{\partial \mathbf{a}_{1}^{\circ T} }{\partial\phi_{1}^{\circ}}+\dot{\theta}_{1}^{\circ}\frac{\partial \mathbf{a}_{1}^{\circ T} }{\partial\theta_{1}^{\circ}}=\dot{\mathbf{a}}_{1}^{\circ T}$, we have
$
(\!\frac{r_{i}^{\circ}\dot{r}_{1}^{\circ}}{r_{1}^{\circ}}\!-\dot{r}_{i}^{\circ}\!)\mathbf{a}_{1}^{\circ T}\!\! -\!\frac{r_{i1}^{\circ}}{r_{1}^{\circ}}\dot{\mathbf{u}}^{\circ T}\!\!\!+r_{i1}^{\circ}\!(\dot{\phi}_{1}^{\circ}\!\frac{\partial \mathbf{a}_{1}^{\circ T} }{\partial\phi_{1}^{\circ}}\!+\dot{\theta}_{1}^{\circ}\!\frac{\partial \mathbf{a}_{1}^{\circ T} }{\partial\theta_{1}^{\circ}}\!)
\!=\!-\dot{r}_{i1}^{\circ}\mathbf{a}_{1}^{\circ T},
$
therefore, (b) holds.

\end{appendices}

\end{document}